\documentclass{llncs}
\usepackage{graphicx} 
\usepackage{amsmath,  amsfonts, amscd, amssymb}
\usepackage{epsfig}
\usepackage{color}
\usepackage{gastex}
\usepackage{subfigure}
\usepackage{comment}
\usepackage{xspace}
\usepackage{booktabs}
\usepackage{alltt}

\usepackage{tikz}
\usetikzlibrary{automata,shapes}

\def\ifshort{\ifx} 

\newcommand{\highlight}[1]{#1}

\newcommand{\figref}{Figure~}

\newcommand{\lang}{\L}

\newcommand{\Parity}[1][]{{\mathsf{parity}_{#1}}}
\newcommand{\Mean}[1][]{\mathsf{mean}_{#1}}
\newcommand{\MP}[1][]{\mathsf{mp}_{#1}}
\newcommand{\Max}[1]{\mathsf{max}_{#1}}

\newcommand{\V}[3]{{\cal V}_{#1}^{#2}(#3)}

\def\fsize{\footnotesize}
\gasset{loopdiam=4,Nw=6,Nh=6,Nmr=6,ilength=3}

\def\g{g}
\def\r{r}
\def\ng{{\bar{g}}}
\def\nr{{\bar{r}}}

\def\sstate{q}
\def\Sstate{Q}
\def\run{\rho}

\newcommand{\atrans}{\Delta}

\def\A{\Sigma}
\def\a{\sigma}
\def\w{w}
\def\L{L\xspace}

\def\d{d\xspace}
\def\D{{\cal D}\xspace}
\def\reward{r}
\def\maps{\rightarrow}

\def\I{I}
\renewcommand\O{O}
\def\outcome{{\cal O}}

\newcommand{\tuple}[1]{(#1)}
\renewcommand{\epsilon}{\varepsilon}

\def\cale{{\cal E}}

\newcommand{\pat}{\omega}
\newcommand{\Paths}{\Omega}
\newcommand{\fpath}{v}

\newcommand{\PA}{1}

\newcommand{\straa}{\pi}
\newcommand{\Straa}{\Pi}
\newcommand{\SA}{S_1}

\newcommand{\SR}{S_{P}}
\newcommand{\ovSA}{\ov{S}_1}
\newcommand{\ovSR}{\ov{S}_P}

\newcommand{\gamegraph}{G}

\newcommand{\va}{\mathsf{V}}

\newcommand{\Prb}{\mu}
\newcommand{\Inf}{\mathrm{Inf}}

\newcommand{\Exp}{\mathbb{E}}
\newcommand{\nats}{\mathbb{N}}
\newcommand{\reals}{\mathbb{R}}
\newcommand{\set}[1]{\{ #1 \}}

\newcommand{\obciach}{\upharpoonright}
\newcommand{\trans}{\delta}
\newcommand{\distr}{{\cal D}}

\newcommand{\Aa}{{\cal A}}
\newcommand{\ov}{\overline}

\newcommand{\wh}{\widehat}

\newcommand{\quasy}{{\sc Quasy}\xspace}

\definecolor{Gray}{rgb}{0.7,0.7,0.7} 

\newcommand\url[1]{{#1}}

\title{Measuring and Synthesizing Systems in Probabilistic
  Environments\thanks{This research was supported by the European
    Union project COMBEST.}}

\title{Measuring and Synthesizing Systems in Probabilistic Environments}

\author{Krishnendu Chatterjee$^1$ \and Thomas A. Henzinger$^{1,2}$ \and\\
  Barbara Jobstmann$^3$ \and Rohit Singh$^4$}

\institute{Institute of Science and Technology Austria (IST Austria)
\and \'{E}cole Polytechnique F\'{e}d\'{e}ral de Lausanne (EPFL),
Switzerland \and CNRS/Verimag, France \and
Indian Institute of Technology (IIT), Bombay}

\begin{document}
\maketitle

\begin{center}
\emph{Technical Report, April 14, 2011}
\end{center}

\begin{abstract}
Often one has a preference order among the different systems that
satisfy a given specification.  Under a probabilistic assumption about
the possible inputs, such a preference order is naturally expressed by
a weighted automaton, which assigns to each word a value, such that a
system is preferred if it generates a higher expected value.  We solve
the following optimal-synthesis problem: given an omega-regular
specification, a Markov chain that describes the distribution of
inputs, and a weighted automaton that measures how well a system
satisfies the given specification under the given input assumption,
synthesize a system that optimizes the measured value.

For safety specifications and measures that are defined by mean-payoff
automata, the optimal-synthesis problem amounts to finding a strategy
in a Markov decision process (MDP) that is optimal for a long-run
average reward objective, which can be done in polynomial time.  For
general omega-regular specifications, the solution rests on a new,
polynomial-time algorithm for computing optimal strategies in MDPs
with mean-payoff parity objectives.  Our algorithm generates optimal
strategies consisting of two memoryless strategies and a counter.
This counter is in general not bounded. To obtain a finite-state
system, we show how to construct an $\epsilon$-optimal strategy with a
bounded counter for any $\epsilon >0$.  Furthermore, we show how to
decide in polynomial time if we can construct an optimal finite-state
system (i.e., a system without a counter) for a given specification.

We have implemented our approach and the underlying algorithms in a
tool that takes qualitative and quantitative specifications and
automatically constructs a system that satisfies the qualitative
specification and optimizes the quantitative specification, if such a
system exists.  We present some experimental results showing optimal
systems that were automatically generated in this way.
\end{abstract}

\section{Introduction}
Building correct and reliable programs is one of the key challenges in
computer science.  Automatic verification and synthesis aims to
address this problem by defining correctness with respect to a
\emph{formal specification}, a mathematical description of the desired
behavior of the system. In automatic verification, we ask if a given
system satisfies a given
specification~\cite{ClarkeE81,QueilleS82,CousotC77}.
The synthesis problem asks to automatically derived a system from a
specification~\cite{Church62,Ramadg89,Pnueli89}.  Traditionally, the
verification and synthesis problem are studied with respect to Boolean
specifications in an adversarial environment: the Boolean (or
qualitative) specification maps each possible behavior of a system to
true or false indicating if this behavior is a desired behavior or
not.  Analyzing a system in an adversarial environment corresponds to
considering the system under the worst-case behavior of the
environment.
In this work we study the verification and synthesis problem for
quantitative objectives in probabilistic environments, which
corresponds to analyzing the system under the average-case behavior of
its environment.

Quantitative reasoning is traditionally used to measure quantitative
properties of systems, such as performance or reliability
(cf.~\cite{DeAlfa97,Haverk98,BaierK08,KNP09a}).  Quantitative
reasoning has also been shown useful in the classically Boolean
contexts of verification
and synthesis 
\cite{BloemCHJ09,Katz10}.
In particular, by {augmenting} a Boolean specifications with a
quantitative specifications, we can measure how ``well'' a system
satisfies the specification.  For example, among systems that respond
to requests, we may prefer one system over another if it responds
quicker, or it responds to more requests, or it issues fewer
unrequested responses, etc.  In synthesis, we can use such measures to
guide the synthesis process towards deriving a system that is, in the
desired sense, ``optimal'' among all systems that satisfy the
specification \cite{BloemCHJ09}.

There are many ways to define a quantitative measure that captures the
``goodness'' of a system with respect to the Boolean specification,
and particular measures can be quite different, but there are two
questions every such measure has to answer: 
(1)~how to assign a quantitative value to one particular behavior of a
system (measure along a behavior) and 
(2)~how to aggregate the quantitative values that are assigned to the
possible behaviors of the system (measure across behaviors).
Recall the response property.  Suppose there is a sequence of requests
along a behavior and we are interested primarily in response time,
i.e., the quicker the system responds, the better.  As measure (1)
along a particular behavior, we may be interested in an average or the
supremum (i.e., worst case) of all response times, or in any other
function that aggregates all response times along a behavior into a
single real value.  The choice of measure (2) across behaviors is
independent: we may be interested in an average of all values assigned
to individual behaviors, or in the supremum, or again, in some other
function.
In this way, we can choose to measure the average (across behaviors)
of average (along a behavior) response times, or the average of
worst-case response times, or the worst case of average response
times, or the worst case of worst-case response times, etc.  Note that
these are the same two choices that appear in weighted automata and
max-plus algebras (cf.~\cite{Droste09,Gaubert97,Cuning79}).

In previous work, we studied various measures (1) along a behavior.
In particular, lexicographically ordered tuples of averages
\cite{BloemCHJ09} and ratios \cite{BloemGHJ09} are of natural interest
in certain contexts.  Alur et al.~\cite{Alur09} consider an automaton
model with a quantitative measure (1) that is defined with respect to
certain accumulation points along a behavior.
However, in all of these cases, for measure (2) only the worst case
(i.e., supremum) is considered.  This comes natural as an extension of
Boolean thinking, where a system fails to satisfy a property if even a
single behavior violates the property.
But in this way, we cannot distinguish between two systems that have
the same worst cases across behaviors, but in one system almost all
possible behaviors exhibit the worst case, while in the other only
very few behaviors do so.  In contrast, in performance evaluation one
usually considers the average case across different
behaviors. 

For instance, consider a resource controller for two clients.  Clients
send requests, and the controller grants the resource to one of them
at a time.  Suppose we prefer, again, systems where requests are
granted ``as quickly as possible.''  Every controller that avoids
simultaneous grants will have a behavior along which at least one
grant is delayed by one step, namely, the behavior along which both
clients continuously send requests.  The best the controller can do is
to alternate between the clients.  However, if systems are measured
with respect to the worst case across different behaviors, then a
controller that always alternates between both clients, independent of
the actual requests, is as good as a controller that tries to grant
all requests immediately and only alternates when both clients request
the resource at the same time.
Clearly, if we wish to synthesize the preferred controller, we need to
apply an average-case measure across behaviors.

In this paper, we present a measure (2) that averages across all
possible behaviors of a system and solve the corresponding synthesis
problem to derive an optimal system.  In synthesis, the different
possible behaviors of a system are caused by different input
sequences.  Therefore, in order to take a meaningful average across
different behaviors, we need to assume a probability distribution over
the possible input sequences.
For example, if on input 0 a system has response time~$r_0$, and on
input 1 response time~$r_1$, and input 0 is twice as likely as
input~1, then the average response time is $(2r_0+r_1)/3$.

The resulting synthesis problem is as follows: given a Boolean
specification~$\varphi$, a probabilistic input assumption~$\mu$, and a
measure that assigns to each system~$M$ a value $\V{\mu}{\varphi}{M}$
of how ``well'' $M$ satisfies $\varphi$ under~$\mu$, construct a
system $M$ such that $\V{\mu}{\varphi}{M}\geq \V{\mu}{\varphi}{M'}$
for all~$M'$.
We solve this problem for qualitative specifications that are given as
$\omega$-automata, input assumptions that are given as finite Markov
chains, and a quantitative specification given as mean-payoff automata
which defines a quantitative language by assigning values to
behaviors.  From the above three inputs we derive a measure that
captures (1)~an average along system behaviors as well as (2)~an
average across system behaviors; and thus we obtain a measure that
induces a value for each system.

To our knowledge this is the first solution of a synthesis problem for
an average-case measure across system behaviors.  Technically the
solution rests on a new, polynomial-time algorithm for computing
optimal strategies in MDPs with mean-payoff parity objectives.  In
contrast to MDPs with mean-payoff objectives, where pure memoryless
optimal strategies exist, optimal strategies for mean-payoff parity
objectives in MDPs require infinite memory.  It follows from our
result that the infinite memory can be captured with a counter, and
with this insight we develop the polynomial time algorithm for solving
MDPs with mean-payoff parity objectives.
\highlight{A careful analysis of the constructed strategies allows us
  to construct, for any $\epsilon>0$, a finite-state system that is
  within $\epsilon$ of the optimal value.
  Furthermore, we present a polynomial-time procedure to decide if
  there exists a finite-state system (system without a counter) that
  achieves the optimal value for a mean-payoff parity specification.
  We show that for MDPs with mean-payoff parity objectives finite
  memory does not help, i.e., either the optimal strategy requires
  infinite memory or there exists a memoryless strategy that also
  achieves the optimal value.  We give a linear program to check if
  there exists a memoryless strategy that is optimal.  }

\paragraph{\bf Related works}  Many formalisms for
quantitative specifications have been considered in the
literature~\cite{Alur09,CCHKM05,CAHS03,CAFHMS06,CDH08,dA98,DiscountingTheFuture,DrosteGastin07,DrosteKR08,KL07};
most of these works (other than~\cite{Alur09,CDH08,dA98}) do not
consider mean-payoff specifications and none of these works focus on
how quantitative specifications can be used to obtain better
implementations for the synthesis problem.  Furthermore, several
notions of metrics for probabilistic systems and games have been
proposed in the literature~\cite{dAMRS07,DGJP99}; these metrics
provide a measure that indicates how close are two systems with
respect to all temporal properties expressible in a logic; whereas our
work uses quantitative specification to compare systems with respect 
to the property of interest.
Similar in spirit but based on a completely different technique, is
the work by Niebert et al.~\cite{NiebertPP08}, 
who group behaviors
into good and bad with respect to satisfying a given LTL specification
and use a CTL$^*$-like analysis specification to quantify over the
good and bad behaviors. This measure of logical properties was used by
Katz and Peled~\cite{Katz10} to guide genetic algorithms to discover
counterexamples and corrections for distributed protocols.
Control and synthesis in the presence of uncertainty has been
considered in several works such as~\cite{Baier,CY90,BdA95}: in all
these works, the framework consists of MDPs to model nondeterministic
and probabilistic behavior, and the specification is a Boolean
specification.  In contrast to these works where the probabilistic
choice represent uncertainty, in our work the probabilistic choice
represent a model for the environment assumption on the input
sequences that allows us to consider the system as a whole.  Moreover,
we consider quantitative objectives.  Parr and
Russel~\cite{Parr97reinforcementlearning} also synthesize strategies
for MDPs that optimize a quantitative objectives. They optimize with
respect to the expected discounted total reward, while we consider
mean-payoff objectives. Furthermore, we allow the user (i) to provide
additionally qualitative (in particular liveness) constraints and (ii)
to specify the qualitative and the quantitative constraints
independent of the MDP.  MDPs with mean-payoff objectives are well
studied.  The books~\cite{FV96,Puterman} present a detailed analysis
of this topic.  We present a solution to a more general condition: the
Boolean combination of mean-payoff and parity condition on MDPs.  We
show that MDPs with mean-payoff parity objectives can be solved in
polynomial time.

\paragraph{\bf Structure of the paper}
Section~\ref{sec:preliminaries} gives the necessary theoretical
background and fixes the notation.  In Section~\ref{sec:measures} we
introduce the problem of measuring systems with respect to
quantitative specifications using several examples, define our new
measure, and show how to compute the value of a system with respect to
this measure.  In Section~\ref{sec:synthesis} we show how to construct
a system that satisfy a qualitative specification and optimize a
quantitative specification with respect to our new measure.  In
Section~\ref{sec:experiments} we present experimental results and we
conclude in Section~\ref{sec:conclusion}.

This paper is an extended and improved version
of~\cite{ChatterjeeHJS10} that includes new theoretical results, more
examples, detailed proofs, and reports on an improved implementation.
We present new theoretical results related to finite-state strategies
for approximating the values in mean-payoff parity MDPs and a
polynomial-time procedure to decide the existence of memoryless
strategy that achieves the optimal value.

\section{Preliminaries}
\label{sec:preliminaries}
\subsection{Alphabet, Words, and Languages}
An \emph{alphabet} $\A$ consists of a finite set of \emph{letters}
$\a\in\A$. \highlight{We often use letters representing assignments to
  a set of Boolean variables~$V$. In this case we write $\A=2^V$,
  i.e., $\A$ is the set of all subsets of $V$, and a letter
  $\a=\set{v_1,\dots,v_n} \in 2^V$ encodes the unique assignment,
  in which all variables in $\a$ are set to true and all other
  variables are set to false.}
A \emph{word} $\w$ over $\A$ is either a \emph{finite} or
\emph{infinite} sequence of letters, i.e., $\w\in\A^*\cup\A^\omega$.
Given a word $\w\in \A^\omega$, we denote by $\w_i$ the letter at
position $i$ of $\w$ and by $\w^i$ the prefix of $\w$ of length $i$,
i.e., $\w^i=\w_1\w_2\dots\w_i$.  We denote by $|\w|$ the length of the
word $\w$, i.e., $|\w^i|=i$ and $|\w|=\infty$, if $\w$ is infinite.  A
\emph{qualitative language} \L is a subset of $\A^\omega$.  A
\emph{quantitative language}~$\L$ \cite{CDH08} is a mapping from the
set of words to the set of reals, i.e., $\L: \A^\omega \maps \reals$.
Note that the characteristic function of a qualitative language~$\L$
is a quantitative language mapping words to~$0$ and~$1$.  Given a
qualitative language~$\L$, we use $\L$ also to denote its
characteristic function.

\subsection{Automata with Parity, Safety, and Mean-Payoff
  Objective}
An \emph{(finite-state) automaton} is a tuple
$A=\tuple{\A,Q,q_0,\atrans}$, where $\A$ is a \emph{alphabet},
$Q$ is a (finite) set of \emph{states}, $q_0\in Q$ is an \emph{initial
  state}, and $\atrans: Q \times \A \to Q$\footnote{Note
  that our automata are deterministic and complete to simplify the
  presentation.} is a \emph{transition function} that maps a state and
a letter to a successor state. 
The \emph{run of $A$ on a word $\w=\w_0\w_1\dots$} is a sequence of
states $\run=\run_0 \run_1\dots$ such that (i) $\run_0=q_0$ and (ii)
for all $0\le i \le |\w|$, $\atrans(\run_i,\w_i)=\run_{i+1})$.

A \emph{parity automaton} is a tuple
$A=\tuple{\tuple{\A,Q,q_0,\atrans},p}$, where
$\tuple{\A,Q,q_0,\atrans}$ is a finite-state automaton and $p:Q\maps
\set{0,1,\dots,d}$ is a \emph{priority function} that maps every state
to a natural number in $[0,d]$ called \emph{priority}.  A parity
automaton $A$ \emph{accepts a word} $\w$ if the least priority of all
states occurring infinitely often in the run $\run$ of $A$ on $\w$ is
even, i.e., $\min_{q\in\Inf(\run)} p(q)$ is even, where
$\Inf(\run)=\set{q\mid \forall i \exists j>i\, \run_j=q}$.  The
\emph{language} of $A$ denoted by $\lang_A$ is the set of all words
accepted by~$A$.
A \emph{safety automaton} is a parity automaton with only
priorities~$0$ and~$1$, and no transitions from priority-$1$ to
priority-$0$ states.
A \emph{mean-payoff automaton} is a tuple
$A=\tuple{\tuple{\A,Q,q_0,\atrans},\reward}$, where
$\tuple{\A,Q,q_0,\atrans}$ is a finite-state automaton and $\reward: Q
\times \A \rightarrow \nats$ is a \emph{reward function} that
associates to each transition of the automaton a \emph{reward} $v
\in\nats$.  A mean-payoff automaton assigns to each word $\w$ the
long-run average of the rewards, i.e., for a word $\w$ let $\run$ be
the run of $A$ on $\w$, then we have
\[
\L_A(\w)= 
\begin{cases}
 \frac{1}{n}\cdot \sum_{i=1}^n \reward(\run_{i},\w_i)  & \mbox{if $\w$ is finite,}\\
 \lim\inf_{n \to \infty} \L_A(w^n) &\mbox{otherwise.}
\end{cases}
\]
Note that $\L_A$ is a function assigning values to words. 

\highlight{
\begin{example}
  \figref\ref{fig:automaton} shows a mean-payoff automaton
  $A=\tuple{\tuple{\A,Q,q_0,\atrans},\reward}$ for words over the
  alphabet
  $\A=2^{\set{\r,\g}}=\set{\set{},\set{r},\set{g},\set{r,g}}$, which
  are all possible assignments to the two Boolean variables~$\r$
  and~$\g$. E.g., the letter $\set{\r}$ means that variable~$\r$ is
  true and all the other variables (in this case only $\g$) are false.
  The automaton has two states~$q_0$ and~$q_1$ represented by
  circles. State~$q_0$ is the initial state, which is indicated by the
  straight arrow from the left.  Transitions are represented by
  directed arrows.  They are labeled with (i) a conjunction of
  literals representing a set of letters and (ii) in parentheses, the
  reward obtained when following this transition. If a variable~$v$
  appears in positive form in a label, then we can take this
  transition only with a letter that includes~$v$. If the variable~$v$
  appear in negated form (i.e.,~$\bar{v}$), then this transition can
  only be taken with letter that do not include~$v$. Note that
  transitions depend only on the signals that appear in their
  labels. E.g., the self-loop on state~$q_0$ labeled with~$\g(1)$
  means that we can move from~$q_0$ to~$q_0$ with any letter that
  includes $\g$, i.e., either with letter $\set{\g}$ or with
  letter~$\set{r,g}$.
  The automaton assigns to each word in $\A^\omega$
  the average reward.  E.g., the run of $A$ on the word
  $(\set{\r}~\set{\r}~\set{\r\g})^\omega$ is $(q_0~q_0~q_1)^\omega$
  and the corresponding sequence of rewards is $(0~0~1)^\omega$ with
  an average reward of $(0+0+1)/3=1/3$.
\end{example}
}

\highlight{
\begin{figure}[t]
  \begin{minipage}[b]{0.5\textwidth}
    \begin{center}
\begin{tikzpicture}[node distance=2.5cm,auto,bend angle=20]
  {\fsize
    \node[state, initial, initial text=] (q0) {$q_0$};
    \node[state, right of=q0] (q1) {$q_1$};
   
    \path[->] (q0)   
    edge [loop above] node {$\g (1)$} (q0)
    edge [loop below] node {$\nr\ng (1)$} (q0)
    edge [bend left] node {$\r\ng(0)$} (q1);
    
    \path[->] (q1) 
    edge [loop above] node {$\ng(0)$} (q1)
    edge [bend left] node {$\g(1)$} (q0);
  }
\end{tikzpicture}

      \caption{Mean-payoff automaton~$A$}
      \label{fig:automaton}
    \end{center}
  \end{minipage}
  \begin{minipage}[b]{0.5\textwidth}
   \begin{center}
\begin{tikzpicture}[node distance=2.5cm,auto,bend angle=20]
  {\fsize
    \node[state, initial, initial text=] (q0) {$\sstate_0$};
    \node[state, right of=q0] (q1) {$\sstate_1$};
   
    \path[->] (q0)   
    edge [loop below] node {$\nr/\ng$} (q0)
   edge [bend left] node {$\r/\g$} (q1);
   
    \path[->] (q1) 
    edge [bend left] node {$/\ng$} (q0);
  }
\end{tikzpicture}

      \caption{Finite-state system~$M$}
      \label{fig:system}
    \end{center}
 \end{minipage}
\end{figure}
}

\subsection{State machines and Specifications}
A \emph{(finite-)state machine} (or \emph{system}) with \emph{input
  signals} $\I$ and \emph{output signals} $\O$ is a tuple
$M=\tuple{\Sstate,\sstate_0,\atrans,\lambda}$, where
$\tuple{\A_I,\Sstate,\sstate_0,\atrans}$ with $\A_I=2^\I$ is a
(finite-state) automaton and $\lambda: \Sstate \times \A_I \maps \A_0$
with $\A_I=2^\I$ and $\A_O=2^O$ is a \emph{labeling function} that
maps every transition in $\atrans$ to an element in $\A_O$.  The sets
$\A_I$ and $\A_O$ are called the \emph{input and the output alphabet
  of $M$}, respectively.  We denote the joint alphabet $2^{\I\cup\O}$
by $\A$.

Given an input word $\w\in\A_I^*\cup\A_I^\omega$, let $\run$ by the
run of $M$ on $\w$, the \emph{outcome} of $M$ on $\w$, denoted
by~$\outcome_M(\w)$, is the word $v\in \A^*\cup\A^\omega$ s.t.\ for
all $0\le i\le|\w|$, $v_i=\w_i \cup \lambda(\run_i,\w_i)$.  Note that
$\outcome_M$ is the function mapping input words to outcomes.  The
\emph{language} of $M$ denoted by $\lang_M$ is the set of
outcomes of $M$ \highlight{on all infinite input word}.

\highlight{
\begin{example}
  Consider the system~$M$ depicted in \figref\ref{fig:system}.
  System~$M$ has one Boolean input variable~$\r$ and one Boolean
  output variables~$\g$. In every step, $M$ reads the value of the
  variable~$\r$ and sets the value of the variable~$\g$.  More
  precisely, $M$ sets $\g$ to false, whenever either~$\r$ is false in
  the current step or~$\g$ have been true in the previous step.
  The input alphabet of $M$ is~$2^{\set{\r}}=\set{\set{},\set{\r}}$
  and its output alphabet is~$2^{\set{\g}}=\set{\set{},\set{\g}}$.
  Recall that all variables that are absent in a letter are set to
  false, e.g., the input letter $\set{}$ means that the value of $\r$
  is false, while $\set{\r}$ refers to $\r$ being true. We again label
  edges with conjunctions of literals.  The conjunction on the left of
  the slash describes a set of input letters, i.e., a set of
  assignments to the input variables. The conjunction on the right
  describes a single output letter, which corresponds to an assignment
  of the output varibles. E.g., the transition from state~$q_1$ to
  state~$q_0$ labeled $/\ng$ means that if the system is in
  state~$q_1$, then it moves to the state~$q_0$ and sets the
  variables~$\g$ to false for any input letter because the conjunction
  for the input variables is empty.

  Consider the input word $\w=\set{\r} \set{\r} \set{} \set{\r}$. The
  outcome of~$M$ of~$\w$ is the combined word $\set{\r\g} \set{\r}
  \set{} \set{\r\g}$. The language of~$M$ are all the infinite words
  generated by arbitrarily concatenating the following three words:
  (i) $\w_1=\set{}$, (ii) $\w_2=\set{r,g}\set{r}$, and (iii)
  $\w_3=\set{r,g}\set{}$, i.e., $\lang_M = (\w_1 | \w_2
  |\w_3)^\omega$.  
\end{example}
}

We analyze state machines with respect to qualitative and quantitative
specifications.
\emph{Qualitative specifications} are qualitative languages, i.e.,
subsets of $\A^\omega$ or equivalently functions mapping words to~$0$
and~$1$.  We consider $\omega$-regular specifications given as safety
or parity automata.  Given a safety or parity automaton $A$ and a
state machine $M$, we say \emph{$M$ satisfies $L_A$} (written
$M\models\L_A$) if $\lang_M\subseteq\lang_A$ or equivalently $\forall
\w \in\A_I^\omega: \lang_A(\outcome_M (\w))=1.$
\emph{A quantitative specification} is given by a quantitative
language $\L$, i.e., a function that assigns values to words.  Given a
state machine $M$, we use function composition to relate $\L$ and $M$,
i.e., $\L \circ \outcome_M$ is mapping every input word $\w$ of $M$ to
the value assigned by $\L$ to the outcome of $M$ on $\w$.  We consider
quantitative specifications given by Mean-payoff automata.

\subsection{Markov Chains and Markov Decision Processes (MDP)} 
A \emph{probability distribution} over a finite
set $S$ is a function $\d:S\maps [0,1]$ such that $\sum_{q\in Q}
\d(q)=1$.  We denote the set of all probabilistic distributions over
$S$ by $\D(S)$.
A \emph{Markov Decision Process (MDP)} $\gamegraph =\tuple{S, s_0, E,
  \SA,\SR,\trans}$ consists of a finite set of \emph{states} $S$, an
\emph{initial state} $s_0\in S$, a set of \emph{edges} $E\subseteq
S\times S$, a partition $\tuple{\SA$, $\SR}$ of the set $S$, and a
probabilistic transition function $\trans$: $\SR \rightarrow
\distr(S)$.  The states in $\SA$ are the {\em player-$\PA$\/} states,
where player~$\PA$ decides the successor state; and the states in
$\SR$ are the {\em probabilistic\/} states, where the successor state
is chosen according to the probabilistic transition function~$\trans$.
{So, we can view an MDP as a game between two players:
  player $\PA$ and a \emph{random player} that plays according
  to~$\trans$.}
We assume that for $s \in \SR$ and $t \in S$, we have $(s,t) \in E$
iff $\trans(s)(t) > 0$, and we often write $\trans(s,t)$ for
$\trans(s)(t)$.  For technical convenience we assume that every state
has at least one outgoing edge.  For a state $s\in S$, we write $E(s)$
to denote the set $\set{t \in S \mid (s,t) \in E}$ of possible
successors.  If the set $\SA=\emptyset$, then $\gamegraph$ is called a
\emph{Markov Chain} and we omit the partition $\tuple{\SA$, $\SR}$
from the definition.
A \emph{$\A$-labeled MDP} $(G,\lambda)$ is an MDP $G$
with a labeling function $\lambda:S\rightarrow \A$ assigning to each
state of $G$ a letter from~$\A$.  We assume that labeled MDPs are
deterministic and complete, i.e., (i) $\forall(s,s'),(s,s'')\in E$,
$\lambda(s')=\lambda(s'') \rightarrow s'=s''$ holds, and (ii) $\forall
s\in S, \a\in\A,\ \exists s'\in S$ s.t.\ $(s,s')\in E$ and
$\lambda(s')=\a$.

\subsection{Plays and strategies}
An infinite path, or a \emph{play}, of the MDP $\gamegraph$ is an
infinite sequence $\pat=s_0 s_1 s_2\ldots$ of states such that
$(s_k,s_{k+1}) \in E$ for all $k \in \nats$.  Note that we use $\pat$
only to denote plays, i.e., infinite sequences of states.  We
use~$\fpath$ to refer to finite sequences of states.
We write~$\Paths$ for the set of all plays, and for a state $s \in S$,
we write $\Paths_s\subseteq\Paths$ for the set of plays starting
at~$s$.
A \emph{strategy} for player~$\PA$ is a function $\straa$: $S^* \SA
\to \distr(S)$ that assigns a probability distribution to all finite
sequences $\fpath \in S^*\SA$ of states ending in a player-1 state.
Player~$\PA$ follows~$\straa$, if she make all her moves according to
the distributions provided by~$\straa$.
A strategy must prescribe only available moves, i.e., for all $\fpath
\in S^*$, $s \in \SA$, and $t \in S$, if $\straa(\fpath s)(t) >
0$, then $(s, t) \in E$.  We denote by $\Straa$ the set of all
strategies for player~$\PA$.
Once a starting state $s \in S$ and a strategy $\straa \in \Straa$ is
fixed, the outcome of the game is a random walk $\pat_s^{\straa}$ for
which the probabilities of every \emph{event} $\Aa \subseteq \Paths$,
which is a measurable set of plays, are uniquely defined.

For a state $s \in S$ and an event $\Aa\subseteq\Paths$, we write
$\Prb_s^{\straa}(\Aa)$ for the probability that a play belongs to
$\Aa$ if the game starts from the state $s$ and player~1 follow the
strategy $\straa$, respectively.  For a measurable function $f:\Paths
\to \reals$ we denote by $\Exp_s^{\straa}[f]$ the \emph{expectation}
of the function $f$ under the probability measure
$\Prb_s^{\straa}(\cdot)$.

Strategies that do not use randomization are called pure.  A player-1
strategy~$\straa$ is \emph{pure} if for all $\fpath \in S^*$ and $s
\in \SA$, there is a state~$t \in S$ such that $\straa(\fpath s)(t) =
1$.  A \emph{memoryless} player-1 strategy depends only on the current
state, i.e., for all $\fpath,\fpath' \in S^*$ and for all $s \in \SA$
we have $\straa(\fpath s) =\straa(\fpath' s)$.  A memoryless strategy
can be represented as a function $\straa$: $\SA \to \distr(S)$.  A
\emph{pure memoryless strategy} is a strategy that is both pure and
memoryless.  A pure memoryless strategy can be represented as a
function $\straa$: $\SA \to S$. 
\highlight{A \emph{pure finite-state strategy} is a strategy that can
  be represent by a finite-state machine
  $M=\tuple{\Sstate,\sstate_0,\atrans,\lambda}$ with input
  alphabet~$\A_I=S$ and output alphabet~$\A_O=S$. The state~$\Sstate$
  represent a set of memory locations with $\sstate_0$ as the initial
  memory content. The transition function $\atrans: \Sstate \times S
  \to \Sstate$ describes how to update the memory while moving to the
  next state in the MDP. The labeling function~$\lambda: \Sstate
  \times S \to S$ defines the moves of Player 1, i.e., for every
  memory location and state of the MDP, it provides a successor state
  in the MDP.  }

\highlight{ 
  \subsection{Resulting Markov chains, recurrence classes, unichain,
    and multichain}
  Given an MDP~$\gamegraph$ and a pure memoryless or finite-state
  strategy~$\pi$, if we restrict~$\gamegraph$ to follow the actions
  suggested in~$\pi$, we obtain a Markov chain. 

  Given a Markov chain~ $\gamegraph =\tuple{S, s_0, E,\trans}$, a
  state~$s\in S$ is called~\emph{recurrent}\footnote{Note that we do
    not distinguish null or positive recurrent states since we only
    consider finite Markov chains.} if the expected number of visits
  to~$s$ is infinite.  Otherwise, the state~$s$ is
  called~\emph{transient}.  A maximal set of recurrent states that is
  closed\footnote{We use the usual definition for closed, i.e., given
    a set~$Y$, a set $X\subseteq Y$ is closed under a
    relation~$R\subseteq Y\times Y$, if forall $x\in X$ and forall
    $y\in Y$, if $(x,y) \in R$, then $y\in X$.} under $E$ is
  called~\emph{recurrence class.}  A Markov chain~$\gamegraph$ is
  \emph{unichain} if it has a single recurrence
  class. Otherwise,~$\gamegraph$ is called \emph{multichain}.}

\subsection{Quantitative Objectives} A quantitative objective is 
given by a measurable function $f:\Paths \to \reals$.  We consider
several objectives based on priority and reward functions.
Given a priority function $p : S \to \set{0,1,\dots, d}$, we defined
the set of plays satisfying the parity objective as $\Paths_p=
\set{\pat \in \Paths \mid \min\big(p(\Inf(\pat))\big) \text{ is
    even}}.$ A \emph{Parity objective} $\Parity[p]$ is the
characteristic function of~$\Paths_p$.
Given a reward function $r:S \to \nats \cup \set{\bot}$, the
\emph{mean-payoff objective} $\Mean[r]$ for a play
$\pat=s_1s_2s_3\ldots$ is defined as $\Mean[r](\pat)=\lim\inf_{n \to
  \infty} \frac{1}{n}\cdot \sum_{i=1}^n r(s_i)$, if for all
$i>0:r(s_i)\ne\bot$, otherwise $\Mean[r](\pat)=\bot$.
Given a priority function $p$ and a reward function $r$ the
\emph{mean-payoff parity objective} $\MP[p,r]$
assigns the long-run average of the rewards if the parity objective is
satisfied; otherwise it assigns $\bot$.  Formally, for a play~$\pat$ we have
\[
\MP[p,r](\pat)=\begin{cases} \Mean[r](\pat) & \mbox{if }
\Parity[p](\pat)=1,\\
\bot & \mbox{otherwise}.
\end{cases}
\]

For a reward function $r:S \to \reals$ the \emph{max
  objective}~$\Max{r}$ assigns to a play the maximum reward that
appears in the play.  Note that since $S$ is finite, the number of
different rewards appearing in a play is finite and hence the maximum
is defined.  Formally, for a play $\pat=s_1s_2s_3\ldots$ we have
$\Max{r}(\pat)=\max \langle r(s_i) \rangle_{i \geq 0}.$ 

\subsection{Values, optimal stratgies, and almost-sure winning states}
Given an MDP $G$, the \emph{value} function $\va_G$ for an objective
$f$ is the function from the state space $S$ to the set $\reals$ of
reals.  For all states $s\in S$, let
$\va_G(f)(s)  = \displaystyle \sup_{\straa \in \Straa}  
\Exp_s^{\straa}[f].$
In other words, the value $\va_G(f)(s)$ is the maximal expectation
with which player~1 can achieve her objective $f$ from state~$s$.  A
strategy $\straa$ is \emph{optimal} from state $s$
for objective $f$ if $\va_G(f)(s)=\Exp_s^{\straa}[f]$.  For parity
objectives, mean-payoff objectives, and max objectives pure memoryless
optimal strategies exist in MDPs.

Given an MDP $G$ and a priority function $p$, we denote by
$W_G(\Parity[p]) = \set{s \in S \mid \va_G(\Parity[p])(s)=1},$ the set
of states with value~1.  These states are called the
\emph{almost-sure} winning states for the player and an optimal
strategy from the almost-sure winning states is called a
almost-sure winning strategy.  The set $W_G(\Parity[p])$ for an MDP
$G$ with priority function $p$ can be computed in $O(d \cdot
n^{\frac{3}{2}})$ time, where $n$ is the size of the MDP $G$ and $d$
is the number of priorities~\cite{CJH03,CJH04}.  For states in
$S\setminus W_G(\Parity[p])$ the parity objective is falsified with
positive probability for all strategies, which implies that for all
states in $S\setminus W_G(\Parity[p])$ the value is less than~1 (i.e.,
$\va_G(\Parity[p])(s)<1$).

\section{Measuring Systems}
\label{sec:measures}

In this section, we start with an example to explain the problem and
introduce our measure. Then, we define the measure formally and show
finally, how to compute the value of a system with respect to the
given measure.

\begin{example}
  \label{ex:controller}
  Recall the example from the introduction, where we consider a
  resource controller for two clients.  Client~$i$ requests the
  resource by setting its request signal $r_i$.  The resource is
  granted to Client~$i$ by raising the grant signal $g_i$. We require
  that the controller guarantees mutually exclusive access and that it
  is fair, i.e., a requesting client eventually gets access to the
  resource.  Assume we prefer controllers that respond quickly.
  \figref\ref{fig:grant_fast} shows a specification that rewards a
  quick response to request $\r_i$.  The specification is given as a
  Mean-payoff automaton that measures the average delay between a
  request $\r_i$ and a corresponding grant $\g_i$.  Recall that
  transitions are labeled with a conjunction of literals and a reward
  in parentheses.  In particular, whenever a request is granted the
  reward is~$1$, while a delay of the grant results in reward~$0$.
  The automaton assigns to each word in $(2^{\set{\r_i,\g_i}})^\omega$
  the average reward.  For instance, the value of the word
  $(\set{\r_i} \set{\r_i,\g_i})^\omega$ is $(0+1)/2=1/2$.  We can take
  two copies of this specification, one for each client, and assign to
  each word in $(2^{\set{\r_1,\r_2,\g_1,\g_2}})^\omega$ the sum of the
  average rewards.  E.g., the word $(\set{r_1,\g_2}
  \set{\r_1,\g_1})^\omega$ gets an average reward of~$1/2$ with
  respect to the first client and reward~$1$ with respect to the
  second client, which sums up to a total reward of~$3/2$.

\begin{figure}[t]
  \begin{minipage}[b]{0.28\textwidth}
    \begin{center}
\begin{tikzpicture}[node distance=2cm,auto,bend angle=20]
  {\fsize
    \node[state, initial, initial text=] (q0) {$q_0$};
    \node[state, right of=q0] (q1) {$q_1$};
   
    \path[->] (q0)   
    edge [loop above] node {$\g_i (1)$} (q0)
    edge [loop below] node {$\nr_i\ng_i (1)$} (q0)
    edge [bend left] node {$\r_i\ng_i(0)$} (q1);
    
    \path[->] (q1) 
    edge [loop above] node {$\ng_i(0)$} (q1)
    edge [bend left] node {$\g_i(1)$} (q0);
  }
\end{tikzpicture}

      \caption{Automaton~$A_i$}
     \label{fig:grant_fast}
    \end{center}
  \end{minipage}
 \begin{minipage}[b]{0.27\textwidth}
    \begin{center}

\begin{tikzpicture}[node distance=2cm,auto,bend angle=20]
  {\fsize
    \node[state, initial, initial text=] (q0) {$\sstate_0$};
    \node[state, right of=q0] (q1) {$\sstate_1$};
   
    \path[->] (q0)   
    edge [bend left] node {$/\g_1\ng_2$} (q1);
  
    \path[->] (q1) 
    edge [bend left] node {$/\ng_1\g_2$} (q0);

    \textcolor{white}{
    \path[->] (q0)   
    edge [loop below] node {$/\g_1\ng_2$} (q0);}
  }
\end{tikzpicture}

    \end{center}
    \caption{System~$M_1$}
    \label{fig:system1}
  \end{minipage}
  \begin{minipage}[b]{0.45\textwidth}
    \begin{center}
\begin{tikzpicture}[node distance=2cm,auto,bend angle=15]
  {\fsize
    \node[state, initial, initial text=] (q0) {$\sstate_0$};
    \node[state, right of=q0] (q1) {$\sstate_1$};
    \node[state, right of=q1] (q2) {$\sstate_2$};
   
    \path[->] (q0) 
    edge [loop above] node {$\nr_2/\g_1\ng_2$} (q0)
    edge [loop below]       node {$\nr_1\r_2/\ng_1\g_2$} (q0)
    edge [bend left]          node {$\r_1\r_2/\g_1\ng_2$} (q1);
  
    \path[->] (q1) 
    edge [bend left]          node {$\nr_1/\ng_1\g_2$} (q0)
    edge [bend left]          node {$\r_1/\ng_1\g_2$} (q2);

    \path[->] (q2)
    edge [bend left]          node {$\r_2/\g_1\ng_2$} (q1);

    \draw [->] (q2) .. controls (3.5,1.3) and (0.5,1.3) .. (q0);
    \node at (2,1.3) {$\nr_2/\g_1\ng_2$};
  }
\end{tikzpicture}

    \end{center}
    \caption{System $M_2$ preferring~$r_1$.}
    \label{fig:system2}
  \end{minipage}
\end{figure}

  Consider the systems $M_1$ and $M_2$ in \figref\ref{fig:system1}
  and~\ref{fig:system2}, respectively.  Transitions are labeled with
  conjunctions of input and output literals separated by a slash.
  System~$M_1$ alternates between granting the resource to Client~1
  and~2.  System~$M_2$ grants the resource to Client~2, if only
  Client~2 is sending requests.  By default it grants the resource to
  Client~1.  If both clients request, then the controller alternates
  between them.  Both systems are correct with respect to the
  functional requirements describe above: they are fair to both
  clients and guarantee that the resource is not accessed
  simultaneously.

  Though, one can argue that System~$M_2$ is better than~$M_1$ because
  the delay between requests and grants is, for most input sequences,
  smaller than the delay in System~$M_1$.  For instance, consider the
  input trace $(\set{\r_2} \set{\r_1})^\omega$.  The response of
  System~$M_1$ is $(\set{\g_1}\set{\g_2}) ^\omega$.  Looking at the
product between the system $M_1$ and the specifications~$A_1$
and~$A_2$ shown in \figref\ref{fig:prod-system1}, we can see that this
results in an average reward of $1$.
Similar, \figref\ref{fig:prod-system2} shows the product of $M_2$,
$A_1$, and~$A_2$. System~$M_2$ responds with $(\set{\g_2}
\set{\g_1})^\omega$ and obtains a reward of $2$.
Now, consider the sequence $(\set{\r_1,\r_2})^\omega$, which is the
worst input sequence the environment can provide. In both systems,
this sequences leads to a reward of $1$, which is the lowest possible
reward.  So $M_1$ and $M_2$ cannot be distinguished with respect to
their worst case behavior.

In order to measure a system with respect to its average behavior, we
aim to average over the rewards obtained for all possible input
sequences.  Since we have infinite sequences, one way to average is
the limit of the average over all finite prefixes.  Note that this can
only be done if we know the values of finite words with respect to the
quantitative specification.  For instance, for a finite-state
machine~$M$ and a Mean-payoff automaton~$A$, we can define the
average as $\V{\oslash}{L_A}{M}:=\lim_{n\to\infty}
\frac{1}{|\A_I|^n}\sum_{w\in\A_I^n} \L_A(\outcome_M(w^n)).$
However, if we truly want to capture the average behavior, we need to
know, how often the different parts of the system are used. This
corresponds to knowing how likely the different input sequences are.
The measure above assumes that all input sequences are ``equally
likely''.  In order to define measures that take the behavior of the
environment into account, we use a probability measure on input words.
In particular, we consider the probability space $(\A_I^\omega,{\cal
  F},\Prb)$ over $\A_I^\omega$, where ${\cal F}$ is the
$\sigma$-algebra generated by the cylinder sets of $\A^\omega$ (which
are the sets of infinite words sharing a common prefix) (in other
words, we have the Cantor topology on $\A_I^\omega$) and $\Prb$ is a
probability measure defined on $(\A^\omega,{\cal F})$.  We use
finite labeled Markov chains to define the probability
measure~$\Prb$.

 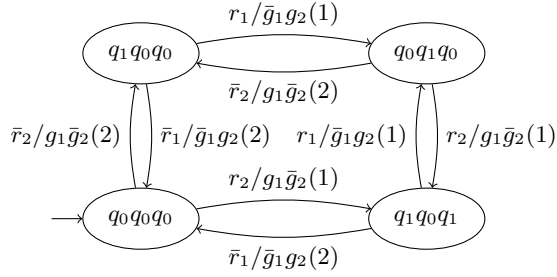
\begin{figure}[t]
      \begin{center}
\begin{tikzpicture}[node distance=2.2cm,auto, bend angle=10]
  {\fsize
   \node[state, initial, initial text=] (q) [shape=ellipse] {${\sstate_0 q_0 q_0}$};
    \node[state, above of=q]           (q0) [shape=ellipse]  {${\sstate_1 q_0 q_0}$};
    \node[state] at(3.8,0)                    (q1) [shape=ellipse]  {${\sstate_1 q_0 q_1}$};
    \node[state, above of=q1]       (q01) [shape=ellipse] {${\sstate_0 q_1 q_0}$};

    \path[->] (q)
    edge [bend left] node {$\nr_2/\g_1\ng_2(2)$} (q0)
    edge [bend left] node {$\r_2/\g_1\ng_2(1)$} (q1);
   
    \path[->] (q0)
    edge [bend left] node {$\r_1/\ng_1\g_2(1)$} (q01)
    edge [bend left] node {$\nr_1/\ng_1\g_2(2)$} (q);

    \path[->] (q1)
    edge [bend left] node {$\r_1/\ng_1\g_2(1)$} (q01)
    edge [bend left] node {$\nr_1/\ng_1\g_2(2)$} (q);

    \path[->] (q01)
    edge [bend left] node {$\nr_2/\g_1\ng_2(2)$} (q0)
    edge [bend left] node {$\r_2/\g_1\ng_2(1)$} (q1);
}
\end{tikzpicture}

      \end{center}
      \caption{Product of $M_1$ with \nobreak{Specification~$A_1$ and~$A_2$.}}
      \label{fig:prod-system1}
\end{figure}
 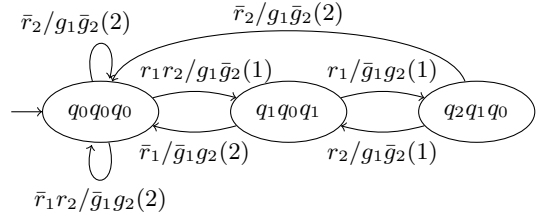
\begin{figure}[t]
      \begin{center}
\begin{tikzpicture}[node distance=2.5cm,auto, bend angle=15]
  {\fsize
    \node[state, initial, initial text=] (q) [shape=ellipse] {${\sstate_0 q_0 q_0}$};
    \node[state, right of=q]             (q0) [shape=ellipse] {${\sstate_1 q_0 q_1}$};
    \node[state, right of=q0]           (q1) [shape=ellipse] {${\sstate_2 q_1 q_0}$};

    \path[->] (q) 
    edge [loop above] node {$\nr_2/\g_1\ng_2(2)\qquad$} (q)
    edge [loop below] node {$\nr_1\r_2/\ng_1\g_2(2)$} (q)
    edge [bend left]    node {$\quad\r_1\r_2/\g_1\ng_2(1)$} (q0);
  
    \path[->] (q0) 
    edge [bend left]          node {$\nr_1/\ng_1\g_2(2)$} (q)
    edge [bend left]          node {$\r_1/\ng_1\g_2(1)$}  (q1);

    \path[->] (q1)
    edge [bend left]          node  {$\r_2/\g_1\ng_2(1)$}(q0);

    \draw [->] (q1) .. controls (4.5,1.3) and (0.5,1.3) .. (q);
    \node at (2.5,1.3) {$\nr_2/\g_1\ng_2(2)$};
  }
\end{tikzpicture}

      \end{center}
      \caption{Product of $M_1$, $A_1$, and $A_2$.}
      \label{fig:prod-system2}
  \end{figure}
\end{example}

\begin{example}
  Recall the controller of Example~\ref{ex:controller}. Assume we know
  We can represent such a behavior by assigning probabilities to the
  events in $\A=2^{\set{\r_1,\r_2}}$.  Assume Client~$1$ sends
  requests with probability $p_1$ and Client~$2$ sends them with
  probability $p_2<p_1$, independent of what has happened before.
  Then, we can build a labeled Markov chain with four states
  $S_p=\set{s_0,s_1,s_2,s_3}$ each labeled with a letter in $\A$,
 i.e., $\lambda(s_0)=\set{}$, $\lambda(s_1)=\set{\r_2}$,
  $\lambda(s_2)=\set{\r_1}$, and $\lambda(s_3)=\set{\r_1,\r_2}$, and the
  following transition probabilities: (i)
 $\trans(s_i)(s_0)=(1-p_1)\cdot(1-p_2)$, (ii)
  $\trans(s_i)(s_1)=(1-p_1)\cdot p_2$, (iii)
  $\trans(s_i)(s_2)=p_1\cdot (1-p_2)$, and (iv)
  $\trans(s_i)(s_3)=p_1\cdot p_2$, for all $i\in\set{0,1,2,3}$.
\end{example}

Once we have a probability measure $\Prb$ on the input sequences and
the associated expectation measure $\Exp$, we can define a
satisfaction relation between systems and specifications and a measure
for a system with respect to a qualitative and a quantitative
specification.

\begin{definition}[Satisfaction]
  Given a state machine $M$ with input alphabet~$\A_I$, a qualitative
  specification $\varphi$, and a probability measure $\Prb$ on
  $(\A_I^\omega,{\cal F})$, we say that \emph{$M$ satisfies $\varphi$
    under $\Prb$} (written $M\models_{\Prb} \varphi$) iff $M$ satisfies
  $\varphi$ with probability~$1$, i.e.,
  $\Exp[\varphi\circ\outcome_M]=1$, where $\Exp$ is the expectation
  measure for $\Prb$. 
\end{definition}

Recall that we use a quantitative specification to describe how
``good'' a system is. Since we aim for a system that satisfies the
given (qualitative) specification and is ``good'' in a given sense, we
define the value of a machine with respect to a qualitative and a
quantitative specification.

\begin{definition}[Value]\label{def:value}
  Given a state machine $M$, a qualitative specification~$\varphi$,
  quantitative specification $\psi$, and a probability measure $\Prb$
  on $(\A_I^\omega,{\cal F})$, the \emph{value of $M$ with respect to
    $\varphi$ and $\psi$ under $\Prb$} is defined as the expectation
  of the function $\psi \circ \outcome_M$ under the probability
  measure $\Prb$ if \emph{$M$ satisfies $\varphi$ under $\Prb$},
  and~$\bot$ otherwise.  Formally, let~$\Exp$ be the expectation
  measure for $\Prb$, then
  \[
  \V{\Prb}{\varphi\psi}{M} :=
  \begin{cases}
    \Exp[\psi \circ \outcome_M] & \mbox{if } M\models_{\Prb} \varphi,\\
    \bot & \mbox{otherwise.}\\
  \end{cases}
  \]
  If $\varphi$ is the set of all words, then we write
  $\V{\Prb}{\psi}{M}$. Furthermore, we say \emph{$M$ optimizes $\psi$
    under $\mu$}, if $\V{\Prb}{\psi}{M} \ge \V{\Prb}{\psi}{M'}$ for
  all systems~$M'$.
\end{definition} 
 
{In Definition~\ref{def:value}, we could also consider the traditional
  satisfaction relation, i.e., $M\models \varphi$.  We have algorithms
  for both notions but we focus on satisfaction under $\mu$, since
  satisfaction with probability~$1$ is the natural correctness
  criterion, if we are given a probabilistic environment assumption.
  Note that for safety specifications the two notions coincide,
  because we assume that the labeled Markov chain defining the input
  distribution is complete.
  \footnote{Recall that a Markov chain is complete, if in every state
    there is an edge for every input value.  Since every edge has a
    positive probability, also every finite path has a positive
    probability and therefore a system violating a safety
    specification will have a value~$\bot$. If the Markov chain is not
    complete (i.e., we are given an input distribution that assigns
    probability~$0$ to some finite input sequences), we require a
    simple pre-processing step that restricts our algorithms to the
    set of states satisfying the safety condition independent of the
    input assumption. This set can be computed in linear time by
    solving a safety game.} %
  For parity specifications, the results in this section would change
  only slightly if we replace $M\models_{\Prb} \varphi$ by $M\models
  \varphi$.  In particular, instead of analyzing a Markov chain with
  parity objective, we would have to analyze an automaton with parity
  objective.  We discuss the the alternative synthesis algorithm in
  the conclusions.}

\begin{lemma}
  \label{lm:mc}
  Given a finite-state machine~$M$, a safety or a parity automaton~$A$, a
  mean-payoff automaton~$B$, and a labeled Markov chain $(G,\lambda_G)$
  defining a probability measure $\Prb$ on $(\A_I^\omega,{\cal F})$,
  we can construct a Markov chain $G'=(S',s_0',E',\delta')$, a reward
  function~$\reward'$, and a priority function~$p'$ such that
\[
\V{\Prb}{L_A,L_B}{M}=
\begin{cases}
  2\cdot\va_{G'}(\Mean[r'])(s_0') & \mbox{if A is a safety automaton,}\\
  2\cdot\va_{G'}(\MP[p',r'])(s_0') & \mbox{otherwise}.\\
\end{cases}
\]
\end{lemma}

\ifshort
Intuitively, we first build the product of~$M$,~$A$,~$B$ (cf.\
\figref\ref{fig:prod-system1}).  Then, $G'$ alternates between (1)
moving according to $\gamegraph$, which means choosing an input value
according to the distribution given by $\gamegraph$, and (2) moving in
$M\times A \times B$ according to the chosen input.  The reward given
by $B$ for this transition is assigned to the intermediate state.  The
priorities are copied from $A$.  The value $\V{\Prb}{\lang_B}{M}$ is
twice the expectation $\va_{G'}(\Mean[r'])(s'_0)$, since we have
introduced $0$-rewards in every second step.  A detailed proof can be
found in the appendix.

\else
\begin{proof} To build $G'$, we first build the product of~$M$,~$A$,
  and~$B$ (cf.\ Figure~\ref{fig:prod-system1}), which is a
  finite-state machine $C=(Q,q_0,\atrans,\lambda)$ augmented with a
  (transition) reward function $\reward: Q \times \A_I \maps \nats$
  and a priority function $p:Q\to \set{0,\dots,d}$.
  Let $\gamegraph=(S,s_0,E,\trans)$, then we construct a Markov chain
  $G'=(S',s_0',E'\cup E'',\trans')$, a reward function
  $\reward':S'\to\nats$, and a priority function
  $p':S'\to\set{0,\dots,d}$ as follows: %
  $S' =Q\times S \times \set{0,1}$, %
  $s_0'=(q_0,s_0,0)$, %
  $E'=\set{((q,s,0),(q,s',1))\mid (s,s')\in E}$, %
  $E''=\set{((q,s,1)(q',s,0))\mid \atrans(q,\lambda_G(s))=q'}$, and

  \[
  \trans'(t)(t')=\begin{cases}
    1 & \mbox{if } (t,t')\in E''\\
    \delta(s,s') & \mbox{if } (t,t')\in E', t=(q,s,0),
    \mbox{ and } t'=(q,s',1)\\
    0 & \mbox{otherwise.}
  \end{cases}
  \]

  In $G'$ every transition of $M\times A$ is split into two parts: in
  the first part, $G'$ chooses the input value according to the
  distribution given by $\gamegraph$.  In the second part, $G'$
  outputs the value from $M$ corresponding to the chosen input.  The
  reward given by $A$ for this transition is assigned to the
  intermediate state, i.e., $\reward'(s')=\reward(q,\lambda_G(s))$ if
  $s'=(q,s,1)$, otherwise $\reward'(s')=0$, and the priorities are
  copied from~$A$, i.e., $p'((q,s,b))=p(q)$.  If~$A$ is a safety
  automaton, we overwrite the rewards function $r'$ to map all states
  $s'\in S'$ with priority $1$ to $\bot$, i.e., $r'(s)=\bot$ if
  $p'(s)=1$. This allows us to ignore the priority function and
  compute the system value based on the mean-payoff value.

Note that we can also compute $M \models_\mu L_A$ and
$\V{\mu}{L_B}{M}$ separately by building two MCs: (1) $G'$ augmented
with a priority function $p'$ and (2) $G''$ augmented with a reward
function $r''$.  Then, $\V{\Prb}{L_A,L_B}{M}=
\va_{G'}(\Mean[r'])(s'_0)$, if $\va_{G''}(\Parity[p''])(s''_0)=1$,
otherwise $\V{\Prb}{L_A,L_B}{M}=\bot$.  Even though, the approach with
two MCs has a better complexity, we constructed a single MC to show
the similarity between verification and synthesis.  \qed
\end{proof}

\fi

\begin{theorem}
  \label{thm:mc}
  Given a finite-state machine~$M$, a parity automaton~$A$, a
  mean-payoff automaton~$B$, and a labeled Markov chain
  $(G,\lambda_G)$ defining a probability measure $\Prb$, we can
  compute the value $\V{\Prb}{L_A,L_B}{M}$ in polynomial time.
  Furthermore, if $(G,\lambda_G)$ defines a uniform input
  distribution, then
  $\V{\oslash}{L_B}{M}=\V{\Prb}{L_B}{M}$\footnote{We can show that
    this measure is invariant under transformations of the computation
    tree.}.
\end{theorem}

\begin{proof}
  Due to Lemma~\ref{lm:mc}, we construct Markov chain $G'$, a reward
  function~$\reward'$, and a priority function~$p'$ such that
  $\V{\Prb}{L_A,L_B}{M}=\va_{G'}(\MP[p',r'])(s_0')$.
  Since $G'$ is a Markov chain, we can compute $W_{G'}(\Parity[p'])$ and
  $\va_{G'}(\Mean[r'])(s'_0)$ in polynomial time~\cite{CJH03,FV96}, and
  $\va_{G'}(\MP[p',r'])(s'_0)=\va_{G'}(\Mean[r'])(s'_0)$ if $s'_0\in
  W_{G'}$, and $\bot$ otherwise.  Note that the value
  $\va_{G'}(\Mean[r'])(s'_0)$ is the sum over all states $s$ of the
  reward at $s$ (i.e., $r'(s)$) times the long-run average frequency of being
  in $s$ (the Cesaro limit of being at $s$~\cite{FV96}).
\qed
\end{proof}

\begin{example}
  Recall the two system~$M_1$ and~$M_2$ (\figref\ref{fig:system1}
  and~\ref{fig:system2}, respectively) and the specification $A$
  (cf.~\figref\ref{fig:grant_fast}) that rewards quick responses.
 The two systems are equivalent wrt the worst case behavior.  Let us
  consider the average behavior: we build a Markov chain $G_{\oslash}$
  that assigns $1/4$ to all events in $2^{\set{\r_1,\r_2}}$. To
  measure $M_1$, we take the product between $G_{\oslash}$ and
  $M_1\times A$ (shown in \figref\ref{fig:prod-system1}).  The product
  looks like the automaton in \figref\ref{fig:prod-system1} with an
  intermediate state for each edge. This state is labeled with the
  reward of the edge. All transition leading to intermediate states
  have probability $1/2$, the other once have probability~$1$.
  So the expectation of being in a state is the same for all four main
  states (i.e., $1/8$) and half of it in the eight intermediate states
  (i.e., $1/16$).
  Four (intermediate) states have a reward of $1$, four have a reward
  of $2$. So we get a total reward of $4\cdot 1/16 + 4\cdot 2\cdot
  1/16=3/4$, and a system value of $1.5$.  This is expected when
  looking at \figref\ref{fig:prod-system1} because each state has two
  inputs resulting in a reward of~$2$ and two inputs with reward~$1$.
  For System~$M_2$, we obtain Markov chain similar to
  \figref\ref{fig:prod-system2} but now the probability of the
  transitions corresponding to the self-loops on the initial state sum
  up to $3/4$.  So it is more likely to state in the initial state,
  then to leave it.  The expectation for being in the states
  $(q_0,q_0,q_0)$,$(q_1,q_0,q_1)$, and $(q_2,q_1,q_0)$ are $2/3$,
  $2/9$, and $1/9$, respectively, and their expected rewards are
  $(2+2+2+1)/4=7/4$, $3/2$, and $3/2$, respectively.  So, the total
  reward of System~$M_2$ is $2/3\cdot 7/4 + 2/9\cdot 3/2 + 1/9\cdot
  3/2=1.67$, which is clearly better than the value of system $M_1$
  for specification $A$.
\end{example}

\section{Synthesizing Optimal Systems}
\label{sec:synthesis}
In this section, we show how to construct a system that satisfies a
qualitative specification and optimizing a quantitative specification
under a given probabilistic environment.  First, we reduce the problem
to finding an optimal strategy in an MDP for a mean-payoff (parity)
objective. Then, we show how to compute such a strategy using end
components and a reduction to max objective. \highlight{In this part,
  we also show how to decide if the given specification can be
  implemented by a finite-state system that is optimal. In case that
  the specification does not permit such an implementation, we show
  how to construct, for every $\epsilon>0$, a finite-state system that
  is $\epsilon$-optimal.} At the end of the section, we provide a
linear program that computes the value function of an MDP with max
objective, which shows that the value function of an MDP with
mean-payoff parity objective can be computed in polynomial time.

\subsection{Reduction to MDP with mean-payoff (parity) objectives}
\begin{lemma}
  \label{lm:mdp}
  Given a safety (resp.~parity) automaton~$A$, a mean-payoff
  automaton~$B$, and a labeled Markov chain $(G,\lambda_G)$ defining a
  probability measure $\Prb$ on $(\A_I^\omega,{\cal F})$, we can
  construct a labeled MDP $(G',\lambda_{G'})$ with $\gamegraph'
  =\tuple{S', s'_0, E', \SA',\SR',\trans'}$, a reward
  function~$\reward'$, and a priority function~$p'$ such that every
  pure strategy $\pi$ that is optimal from state $s'_0$ for the
  objective $\Mean[r']$ (resp. $\MP[p',r']$)  and for which
  $\Exp_{s'_0}^\pi(\Mean[r']) \ne\bot$ (resp.
  $\Exp_{s'_0}^\pi(\MP[p',r']) \ne\bot$) corresponds to a state
  machine~$M$ that satisfies $L_A$ under $\mu$ and optimizes $L_B$
  under $\mu$.
\end{lemma}

The construction of~$G'$ is very similar to the construction used in
Lemma~\ref{lm:mc}.  Intuitively, $G'$ alternates between mimicking a
move of $G$ and mimicking a move of $A\times B \times C$, where $C$ is
an automaton with $|\A_O|$-states that pushes the output labels from
transitions to states, i.e., the transition function~$\delta_C$ of~$C$
is the largest transition function s.t.\ $\forall s,s',\a,\a':
\delta_C(s,\a)=\delta_C(s',\a') \to \a=\a'$. Priorities~$p'$ are again
copied from $A$ and rewards~$r'$ from~$B$. The labels for
$\lambda_{G'}$ are either taken from $\lambda_G$ (in intermediate
state) or they correspond to the transitions taken in~$C$.  Every pure
strategy in $G'$ fixes one output value for every possible input
sequence.  The construction of the state machine depends on the
structure of the strategy.  For pure memoryless strategies, the
construction is straight forward.  At the end of this section, we
discuss how to deal with other strategies.

The following theorem follows from Lemma~\ref{lm:mdp} and the fact
that MDPs with mean-payoff objective have pure memoryless optimal
strategies and they can be computed in polynomial time
(cf.~\cite{FV96}).

\begin{theorem}
  Given a safety automaton~$A$, a mean-payoff automaton~$B$, and a
  labeled Markov chain $(G,\lambda_G)$ defining a probability measure
  $\Prb$, we can construct a finite-state machine~$M$ (if one exists) 
  in polynomial time 
  that satisfies\footnote{Recall that for safety specification
    $M\models_{\mu} L_A$ and $M\models L_A$ coincide.}  $L_A$ under
  $\mu$ and optimizes $L_B$ under $\mu$. 
\end{theorem}

\newcommand\myparagraph[1]{\paragraph{\bf #1}}

\subsection{ MDPs with mean-payoff parity objectives}
\label{sec:meanpayoffparity}
It follows from Lemma~\ref{lm:mdp} that if the qualitative specification
is a parity automaton, along with the Markov chain for probabilistic 
input assumption, and mean-payoff automata for quantitative specification, 
then the solution reduces to solving MDPs with mean-payoff parity 
objective. 
In the following we provide an algorithmic solution of MDPs with
mean-payoff parity objective.
We first present few basic results on MDPs.


\emph{End components} in MDPs~\cite{deAlfaro97,CY90} play a role
equivalent to closed recurrent sets in Markov chains.  Given an MDP
$\gamegraph =\tuple{S, s_0, E, \SA,\SR,\trans}$ , a set $U \subseteq
S$ of states is an \emph{end component}~\cite{deAlfaro97,CY90} if $U$
is $\trans$-closed (i.e., for all $s \in U \cap \SR$ we have $E(s)
\subseteq U$) and the sub-game of $G$ restricted to $U$ (denoted
$G \obciach U$) is strongly connected.  We denote by $\cale(G)$ the
set of end components of an MDP $G$.  The following lemma states that,
given any strategy (memoryless or not), with probability 1 the set of
states visited infinitely often along a play is an end component.
This lemma allows us to derive conclusions on the (infinite) set of
plays in an MDP by analyzing the (finite) set of end components in the
MDP.

\begin{lemma}\label{lemm:end-component}
  \cite{deAlfaro97,CY90} Given an MDP $G$, for all states $s \in S$
  and all strategies $\straa \in \Straa$, we have
  $\Prb_s^{\straa}(\set{\pat \mid \Inf(\pat) \in \cale(G)})=1$.
\end{lemma}
 
For an end component $U \in \cale(G)$, consider the memoryless
strategy $\straa_U$ that plays in any state $s$ in $U \cap \SA$ all
edges in $E(s) \cap U$ uniformly at random.  Given the strategy
$\straa_U$, the end component $U$ is a closed connected recurrent set
in the Markov chain obtained by fixing $\straa_U$.

\begin{lemma}\label{lemm:end-component1}
  Given an MDP $G$ and an end component $U \in \cale(G)$, the strategy
  $\straa_U$ ensures that for all states $s \in U$, we have
  $\Prb_s^{\straa_U}(\set{\pat \mid \Inf(\pat)=U})=1$.
\end{lemma}

It follows that the strategy $\straa_U$ ensures that from any
starting state $s$, any other state $t$ is reached in finite time
with probability~1.  From Lemma~\ref{lemm:end-component1} we can
conclude that in an MDP the value for mean-payoff parity objectives
can be obtained by computing values for end-components and then
applying the maximal expectation to reach the values of the end
components.

\begin{lemma}\label{lemm:key}
  Consider an MDP $G$ with state space $S$, a priority function $p$,
  and reward function $r$ such that (a)~$G$ is an end-component (i.e.,
  $S$ is an end component) and (b)~the minimum priority in~$S$ is
  even.  Then the value for mean-payoff parity objective for all
  states coincide with the value for mean-payoff objective, i.e., for
  all states $s$ we have $\va_G(\MP[p,r])(s)=\va_G(\Mean[r])(s)$.
\end{lemma}
\begin{proof}
  We consider two pure memoryless strategies: one for the mean-payoff
  objective and one for reaching the minimum priority objective
  and combine them to produce the value for mean-payoff parity
  objective.  Consider a pure memoryless optimal strategy $\straa_m$
  for the mean-payoff objective; and the strategy $\straa_S$ is a pure
  memoryless strategy for the stochastic shortest path to reach the
  states with the minimum priority (and the priority is even).
  Observe that under the strategy $\straa_S$ we obtain a Markov chain
  such that every closed recurrent set in the Markov chain contains
  states with the minimum priority, and hence from any state $s$ a
  state with the minimum priority (which is even) is reached in
  finite time with probability~1.  The mean-payoff value for all
  states $s \in S$ is the same: if we fix the memoryless strategy $\straa_u$ 
  that chooses all successors uniformly at random, then we get a Markov 
  chain as the whole set $S$ as a closed recurrent set, and hence from
  all states $s\in S$ any state $t \in S$ is reached in finite time
  with probability~1, and hence the mean-payoff value at $s$ is at
  least the mean-payoff value at $t$.  It follows that for all $s,t\in
  S$ we have $\va(\Mean[r])(s)=\va(\Mean[r])(t)$, and let us denote
  the uniform value by $v^*$. The strategy $\straa_m$ is a pure
  memoryless strategy and once it is fixed we obtain a Markov chain
  and the limit of the average frequency of the states exists and
  since $\straa_m$ is optimal it follows that for all states $s \in S$
  we have
  \[
  \lim_{n \to \infty} \frac{1}{n} \cdot \sum_{i=1}^n
  \Exp_s^{\straa_m}[r(\theta_i)] = v^*,
   \] 
  where $\theta_i$ is the random variable for the $i$-th state of a
  path.  Hence the strategy $\straa_m$ ensures that for any
  $\epsilon>0$, there exists $j(\epsilon) \in \nats$ such that if
  $\straa_m$ is played for any $\ell \geq j(\epsilon)$ steps then the
  expected average of the rewards for $\ell$ steps is within
  $\epsilon$ of the mean-payoff value of the MDP, i.e., for all $s \in
  S$, for all $\ell \geq j(\epsilon)$ we have
  \[
  \frac{1}{\ell}\cdot \sum_{i=0}^\ell \Exp_s^{\straa_m}[r(\theta_i)]
  \geq v^*-\epsilon.
   \]
  Let $\beta$ be the maximum absolute value of the rewards.  The
  optimal strategy for mean-payoff objective is played in rounds, and
  the strategy for round $i$ is as follows:
  \begin{enumerate}
  \item \emph{Stage 1.} First play the strategy $\straa_S$ till the
    minimum priority state is reached.
  \item \emph{Stage 2.} Let $\epsilon_i=1/i$.  If the game was in the
    first stage in this ($i$-th round) for $k_i$ steps, then play the
    strategy $\straa_m$ for $\ell_i$ steps such that $\ell_i \geq
    \max\set{j(\epsilon_i), i \cdot k_i\cdot \beta}$.  This ensures
    that the expected average of the rewards in round $i$ is at least

    \[
    \begin{array}{rcl}
      \displaystyle
      \frac{\ell_i\cdot(v^*-\epsilon_i)}{k_i + \ell_i} & = & 
      \displaystyle 
      \frac{(\ell_i+k_i)\cdot v^*-k_i\cdot v^* - \ell_i\cdot\epsilon_i} 
      {k_i + \ell_i} \\[2ex]
      & \geq & 
     \displaystyle 
      v^*-\frac{\ell_i\cdot\epsilon_i + k_i\cdot v^*}{\ell_i+k_i} \\[2ex]
     & \geq & 
      \displaystyle 
      v^*-\frac{\ell_i\cdot\epsilon_i + k_i\cdot \beta}{\ell_i+k_i} 
      \quad (\text{since } v^* \leq \beta)\\[2ex]
      & \geq & 
      \displaystyle 
      v^* -\frac{\ell_i \cdot \epsilon_i + k_i \cdot \beta}{\ell_i}\\[2ex]
      & = & 
      \displaystyle 
      v^* - \epsilon_i - \frac{k_i\cdot \beta}{\ell_i} \\[2ex]
      & \geq & 
      \displaystyle 
      v^*-\frac{1}{i} - \frac{1}{i}= v^* -\frac{2}{i}.
    \end{array}
    \]
Then the strategy proceeds to round $i+1$.
  \end{enumerate}

  The strategy ensures that there are infinitely many
  rounds, and hence the minimum priority that is visited infinitely
  often with probability~1 is the minimum priority of the end
  component (which is even).  This ensures that the parity objective
  is satisfied with probability~1.  The above strategy ensures that
  the value for the mean-payoff parity objective is
  \[
  \lim\inf_{i \to \infty} (v^*-\frac{2}{i}) = v^*.
  \]
  This completes the proof.
  \qed
\end{proof}

Lemma~\ref{lemm:key} shows that in an end component if the minimum priority
is even, then the value for mean-payoff parity and mean-payoff
objective coincide if we consider the sub-game restricted to the end
component.
\highlight{The strategy constructed in Lemma~\ref{lemm:key} requires infinite memory 
and in the following lemma we show that for all $\epsilon>0$,
the $\epsilon$-approximation can be achieved with finite memory strategies.

\begin{lemma}\label{lemm:key-approx}
  Consider an MDP $G$ with state space $S$, a priority function $p$,
  and reward function $r$ such that (a)~$G$ is an end-component (i.e.,
  $S$ is an end component) and (b)~the minimum priority in~$S$ is
  even.  Then for all $\epsilon>0$ there is a finite-memory strategy 
  $\straa_\epsilon$ for which the mean-payoff parity objective value for all 
  states is within $\epsilon$ of the value for the mean-payoff objective,
  i.e., for all states $s$ we have $\Exp_s^{\straa_\epsilon}[\MP[p,r]]\geq 
  \va_G(\Mean[r])(s)-\epsilon$.
\end{lemma}
\begin{proof}
  The proof of the result is similar as the proof of Lemma~\ref{lemm:key} and 
  the key difference is that the Stage~1 and Stage~2 strategies will be played 
  for a fixed number of rounds, depending on $\epsilon>0$, but will not vary
  across rounds.
  Fix $\epsilon>0$, and we show how to construct a finite-memory strategy 
  to achieve $2\cdot \epsilon$ approximation. As $\epsilon>0$ is arbitrary,
  the desired result will follow. 
  As in Lemma~\ref{lemm:key} we consider the two pure memoryless strategies: 
  one for the mean-payoff objective and one for reaching the minimum priority objective
  and combine them to produce the approximation of the value for mean-payoff parity
  objective.  Consider a pure memoryless optimal strategy $\straa_m$
  for the mean-payoff objective; and the strategy $\straa_S$ is a pure
  memoryless strategy for the stochastic shortest path to reach the
  states with the minimum priority (and the priority is even).
  As in Lemma~\ref{lemm:key} we observe that under the strategy $\straa_S$ 
  we obtain a Markov chain such that every closed recurrent set in the Markov chain contains
  states with the minimum priority, and hence from any state $s$ a
  state with the minimum priority (which is even) is reached in
  finite time with probability~1.  
  Let $n$ be the number of states of the end component, and let $\eta$ be the 
  minimum positive transition probability in the end component. 
  The strategy $\straa_S$ ensures that from all states $s$ there is a path 
  to the minimum even priority state in the graph of the Markov chain, 
  and the path is of length at most $n$. 
  Hence the strategy $\straa_S$ ensures that from all states $s$ a minimum 
  priority state is reached within $n$ steps with probability at least 
  $\eta^n$ (we will refer this as Property~1 later in the proof).
  As shown in Lemma~\ref{lemm:key} the mean-payoff value for all states 
  $s \in S$ is the same: for all $s,t\in S$ we have $\va(\Mean[r])(s)=\va(\Mean[r])(t)$, 
  and let us denote the uniform value by $v^*$. The strategy $\straa_m$ is a pure
  memoryless strategy and once it is fixed we obtain a Markov chain
  and the limit of the average frequency of the states exists and
  since $\straa_m$ is optimal it follows that for all states $s \in S$
  we have
  \[
  \lim_{n \to \infty} \frac{1}{n} \cdot \sum_{i=1}^n
  \Exp_s^{\straa_m}[r(\theta_i)] = v^*,
   \] 
  where $\theta_i$ is the random variable for the $i$-th state of a
  path.  Hence the strategy $\straa_m$ ensures that for any
  $\epsilon_1>0$, there exists $j(\epsilon_1) \in \nats$ such that if
  $\straa_m$ is played for any $\ell \geq j(\epsilon_1)$ steps then the
  expected average of the rewards for $\ell$ steps is within
  $\epsilon_1$ of the mean-payoff value of the MDP, i.e., for all $s \in
  S$, for all $\ell \geq j(\epsilon_1)$ we have
  \[
  \frac{1}{\ell}\cdot \sum_{i=0}^\ell \Exp_s^{\straa_m}[r(\theta_i)]
  \geq v^*-\epsilon_1.
   \]
  Let $\beta$ be the maximum absolute value of the rewards.  The
  finite-memory $2\cdot\epsilon$-optimal strategy for the mean-payoff parity 
  objective is played in rounds, but in contrast to Lemma~\ref{lemm:key} 
  in every round the same strategy is played.
  The strategy for a round is as follows:
  \begin{enumerate}
  \item \emph{Stage 1.} First play the strategy $\straa_S$ for $n$ steps.
  \item \emph{Stage 2.} Play the strategy $\straa_m$ for $\ell$ steps such that 
    $\ell \geq \max\set{j(\epsilon), \frac{1}{\epsilon} \cdot n\cdot \beta}$.  
    This ensures that the expected average of the rewards in a round
    is at least

    \[
    \begin{array}{rcl}
      \displaystyle
      \frac{\ell\cdot(v^*-\epsilon)}{n+ \ell} & = & 
      \displaystyle 
      \frac{(\ell+n)\cdot v^*-n\cdot v^* - \ell\cdot\epsilon} 
      {n + \ell} \\[2ex]
      & \geq & 
     \displaystyle 
      v^*-\frac{\ell\cdot\epsilon + n\cdot v^*}{\ell+n} \\[2ex]
     & \geq & 
      \displaystyle 
      v^*-\frac{\ell\cdot\epsilon + n\cdot \beta}{\ell+n} 
      \quad (\text{since } v^* \leq \beta)\\[2ex]
      & \geq & 
      \displaystyle 
      v^* -\frac{\ell \cdot \epsilon + n\cdot \beta}{\ell}\\[2ex]
      & = & 
      \displaystyle 
      v^* - \epsilon - \frac{n\cdot \beta}{\ell} \\[2ex]
      & \geq & 
      \displaystyle 
      v^*-\epsilon -\epsilon = v^* -2\cdot \epsilon.
    \end{array}
    \]
    Then the strategy proceeds to the next round.
  \end{enumerate}
  The above strategy is a finite-memory strategy as it needs to remember 
  the number $n$ for first stage and the number $\ell$ for second stage.
  The above strategy ensures that the value for the mean-payoff objective is
  at least $v^* -2\cdot \epsilon$. 
  To complete the proof that the strategy is a $2\cdot \epsilon$ optimal 
  strategy we need to show that the parity objective is satisfied with 
  probability~1.
  We call a round a \emph{success} is a minimum even priority state is visited.
  Hence we need to argue that with probability~1 there are infinitely many 
  success rounds.
  Every round is a success with probability at least $\alpha=\eta^n>0$ (as 
  by Property~1 the strategy $\straa_S$ ensures that a minimum priority 
  state is visited with probability at least $\alpha$ in $n$ steps). 
  For round $i$, the probability that there is no success round after round $i$
  is $\lim_{k \to \infty} \alpha^k=0$. 
  Since the countable union of measure zero events has measure zero, it follows 
  that for any round $i$, the probability that there is no success round after 
  round $i$ is zero. It follows that the probability that there are infinitely
  many success rounds is~1, i.e., the parity objective is satisfied with 
  probability~1.
  This completes the proof.
  \qed
\end{proof}
}

\highlight{ In the following we show that if a system can achieve the
  optimal value with a pure finite-state strategy, then it can achieve
  the optimal value also with a pure memoryless strategy.  }

\begin{lemma}\label{lemm:memoryless}
\highlight{
  Consider an MDP $\gamegraph =\tuple{S, s_0, E, \SA,\SR,\trans}$, a
  priority function $p$, and reward function $r$ such that (a)~$S$ is
  an end component and (b)~the minimum priority in~$S$ is even.
  If there exists an optimal pure finite-state strategy $\pi$, then
  there exists an optimal pure memoryless strategy~$\pi'$.}  
\end{lemma}

\begin{proof}
\highlight{
  Let $M$ be the Markov chain obtained by fixing the strategy in
  $\gamegraph$ to $\pi$, i.e., $M$ is the synchronous product of
  $\gamegraph$ and a finite-state system describing $\pi$.
  Since the mean-payoff parity objective is prefix-independent and~$S$
  is an end-component (i.e., all states can reach each other with
  probability~$1$), all recurrence classes in~$M$ have the same
  mean-payoff parity value.
  Therefore we can construct a finite-state strategy~$\hat{\pi}$ such
  that the Markov chain~$\hat{M}$ obtained by fixing the strategy in
  $\gamegraph$ to $\hat{\pi}$ has a single recurrence class.  Let
  $\hat{C}$ be the single recurrence class of~$\hat{M}$ and let
  $\hat{C}|_{\gamegraph}$ be the set of states in~$\gamegraph$ that
  appear in~$\hat{C}$. We know that $\min(p(\hat{C}|_{\gamegraph}))$
  is even.  Let $C_1,\dots, C_k$ be the \emph{component recurrence
    classes} that arise if we fix an optimal pure memoryless strategy
  for the mean-payoff objective in~$\gamegraph$ restricted
  to~$\hat{C}|_{\gamegraph}$.
  Since $\hat{\pi}$ is an optimal strategy, $\hat{C}$ and its
  component recurrence classes $C_1,\dots, C_k$ have the same
  mean-payoff value.  Otherwise, assume there exists some $C_i$ that
  has a higher value, then an infinite-state strategy that alternates
  between playing a strategy that ensures~$C_i$ and a strategy to
  reach the minimal priority (cf.~proof of Lemma~\ref{lemm:key}) would
  achieve a higher mean-payoff parity value, which contradicts the
  assumption that~$\hat{\pi}$ is an optimal strategy. Similarly, if
  some~$C_i$ has a lower value, then removing~$C_i$ would again result
  in a better strategy.
  If there is a recurrence class $C_i$ such that $\min(p(C_i))$ is
  odd, then we can ignore $C_i$ in $\hat{C}$ without changing the
  value.
  Finally, assume there are two component recurrent classes $C_1$ and
  $C_2$ such that $\min(p(C_1))$ and $\min(p(C_2))$ is even, then we
  can ignore one of them without changing the payoff value.
  From these properties, it follows that if there exists an optimal
   finite-state strategy $\pi$, then  there exists a recurrence class
  $C_i$ s.t.~the minimal priority is even and the mean-payoff value is
  the same as the mean-payoff value of $\pi$.  The desired pure
  memoryless strategy $\pi'$ enforces the recurrence class $C_i$ by
  playing a strategy to stay within $C_i$ for states  in $C_i$ and for
  all states outside of $C_i$ it plays a pure memoryless almost-sure
  winning strategy to reach~$C_i$.
}
\end{proof}

\highlight{\subsection{Algorithm based on linear programming}}

\paragraph{\bf Computing best end-component values}  
We first compute a set $S^*$ such that every end component $U$ with
$\min(p(U))$ is even is a subset of $S^*$.  We also compute a function
$f^*:S^* \to \reals^+$ that assigns to every state $s \in S^*$ the
value for the mean-payoff parity objective that can be obtained by
visiting only states of an end component that contains $s$.  The
computation of $S^*$ and $f^*$ is as follows:
\begin{enumerate}
\item $S^*_0$ is the set of maximal end-components with priority~0 and
  for a state $s \in S^*_0$ the function $f^*$ assigns the mean-payoff
  value when the sub-game is restricted to $S^*_0$ (by
  Lemma~\ref{lemm:key} we know that if we restrict the game to the
  end-components, then the mean-payoff values and mean-payoff parity
  values coincide);
\item for $i\geq 0$, let $S^*_{2i}$ be the set of maximal end
  components with states with priority $2i$ or more and that contains
  at least one state with priority~$2i$, and $f^*$ assigns the
  mean-payoff value of the MDP restricted to the set of end components
  $S^*_{2i}$.
\end{enumerate}
The set $S^* = \bigcup_{i=0}^{\lfloor d/2 \rfloor} S^*_{2i}$.  This
procedure gives the values under the end-component consideration.
\highlight{In the following, we show how to check if an end-component
  has a pure memoryless strategy that achieves the optimal value.}

\highlight{
\paragraph{\bf Checking end-component for memoryless strategy}
\label{sec:memoryless}

Let $U\in S^*$ be a maximal end-component with a minimal even
priority, as computed in the previous section.  Without loss of
generality we assume that the MDP is bipartite, i.e., player-1 states
and probabilistic states strictly alternate along every path.  Let
$E_1= E \cap (S_1 \times S_P)$ be the set of player-1 edges, i.e., the
set of edges starting from a player-1 state.
The mean-payoff value of an end-component can be computed using the
following linear program for MDPs with unichain strategies
(cf.~\cite{Puterman,deAlfaro97}):

\newcommand\x[2]{{x_{(#1,#2)}}}

\begin{eqnarray}
  \text{maximize }&&
  \sum_{(s,t)\in E_1} \;\x{s}{t} \cdot (r(s) +
  r(t))\label{eqn:objective}\\[4mm]
  \text{subject to }
  &&\sum_{(s,t)\in E_1} \; \x{s}{t} = 1\\[4mm]
  &&\forall_{s\in S_1} \sum_{t\in S_P, (s,t)\in E}\; \x{s}{t}= 
  \sum_{(s',t')\in E_1} \; \x{s'}{t'} \cdot \trans(t',s)
  \label{eqn:constraints}
\end{eqnarray}

The program has one variable $\x{s}{t}$ for every outgoing edge of a
player-1 state. Intuitively, $\x{s}{t}$ represents the frequency of
being in state~$s$ and choosing the edge to state~$t$.  Note that
all states $s,t$ such that $\x{s}{t} >0$ belong to a recurrence class.
In order to check if there exists an optimal pure memoryless strategy
in~$U$, we call a modified version of the linear program above for
every even priority~$d$. In particular, we add the following
additional constraints:
\begin{equation}
\label{eq:no_smaller_priority}
   \forall_{s\in S_1} \forall_{t\in S_P: (s,t)\in E} \;
    \x{s}{t} = 0 \mbox{ if } p(s)<d \text{ or }p(t) < d
\end{equation}
It requires that in the resulting recurrence class no priority small
than~$d$ is visited. To ensure that the resulting recurrence class
includes at least on state with priority~$d$, we add the following
term to the objective function (Eqn.~\ref{eqn:objective}).
\begin{equation}
\label{eq:minimum_even}
  \sum_{(s,t)\in E_1 \text{ s.t.\ } p(s)=d \text{ or } p(t)=d} \;
  \x{s}{t}
\end{equation}
Finally, let $v$ be the mean-payoff value for~$U$ obtained by solving
the linear program with Eqn.~\ref{eqn:objective}
to~\ref{eqn:constraints}. If there exists an even priority~$d$ such
that the modified linear program (Eqn.~\ref{eqn:objective}
to~\ref{eq:minimum_even}) has a value strictly greater than $v$, then
there exists a pure memoryless strategy in~$S$ that achieves the
optimal value.  If the value of the linear program is strictly greater
than~$v$, then there exists a witness priority~$d$ and a corresponding
edge $(s,t)\in E_1$ such that $\x{s}{t}$ in Eqn.~\ref{eq:minimum_even}
has a positive value.}

\highlight{In order to compute the maximal reachability expectation we
  present the following reduction.}

\paragraph{\bf Transformation to MDPs with max objective}
Given an MDP $\gamegraph =\tuple{S, s_0, E, \SA,\SR,\trans}$ with a
positive reward function $r:S\to \reals^+$ and a priority function
$p:S\to\set{0,\dots,d}$, and let $S^*$ and $f^*$ be the output of the
above procedure.  We construct an MDP $\ov{\gamegraph} =\tuple{\ov{S},
  s_0, \ov{E}, \ovSA,\SR,\trans}$ with a reward function $\ov{r}$ as
follows: $\ov{S} =S \cup \wh{S}^*$ (i.e., the set of states consists
of the state space $S$ and a copy $\wh{S}^*$ of $S^*$), {$\ov{E}= E
  \cup \set{(s,\wh{s}) \mid s \in S^* \cap \SA \mbox{and $\wh{s}$ is
      the copy of $s$ in } \wh{S^*}} \cup \set{(\wh{s},\wh{s}) \mid
    \wh{s} \in \wh{S}^*}$} (i.e., along with edges $E$, for all
player~1 states $s$ in $S^*$ there is an edge to its copy $\wh{s}$ in
$\wh{S}^*$, and all states in $\wh{S}^*$ are absorbing states),
$\ovSA=\SA \cup \wh{S}^*$, $\ov{r}(s)=0$ for all $s \in S$ and
{$\ov{r}(\wh{s})=f^*(s)$, where $\wh{s}$ is the copy of $s$.}  We
refer to this construction as \emph{max} conversion.  The relationship
between $\va_G(\MP[p,r])$ and $\va_{\ov{G}}(\Max{\ov{r}})$ can be
established as follows.
\begin{enumerate}
\item Consider a strategy $\straa$ in $G$. 
  If an end component $U$ is visited infinitely often, and 
  $\min(p(U))$ is odd, then the payoff is $\bot$, and 
  if $\min(p(U))$ is even, then the maximal payoff achievable for the
  mean-payoff parity objective is upper bounded by the payoff of 
  the mean-payoff objective (which is assigned by $f^*$).
  It follows that for all $s \in S$ we have 
  $\va_G(\MP[p,r])(s) \leq\va_{\ov{G}}(\Max{\ov{r}})(s).$
  
\item Let $\ov{\straa}$ be a pure memoryless optimal strategy for the
  objective $\Max{\ov{r}}$ in $\ov{G}$.  We fix a strategy $\straa$ in
  $G$ as follows: if at a state $s \in S^*$ the strategy $\ov{\straa}$
  chooses the edge $(s,\wh{s})$, then in $G$ on reaching $s$, the
  strategy $\straa$ plays according to the strategy of an winning end
  component that ensures the mean-payoff value (as shown in
  Lemma~\ref{lemm:key}), otherwise $\straa$ follows $\ov{\straa}$.  It
  follows that for all $s \in S$ we have
  $\va_G(\MP[p,r])(s) \geq\va_{\ov{G}}(\Max{\ov{r}})(s).$
\end{enumerate} 
It follows that for all $s \in S$ we have 
$\va_G(\MP[p,r])(s) =\va_{\ov{G}}(\Max{\ov{r}})(s).$
{In order to solve $\ov{G}$ with the objective
  $\Max{\ov{r}}$, we set up the following linear program and solve it
  with a standard LP solver (e.g., \cite{glpk}).} 

\paragraph{\bf Linear programming for the max objective in $\ov{G}$} 
 The following linear program characterizes the value
function $\va_{\ov{G}}(\Max{\ov{r}})$.  Observe that we have already
restricted ourselves to the almost-sure winning states
$W_G(\Parity[p])$, and below we assume $W_G(\Parity[p])=S$.  For all
$s \in \ov{S}$ we have a variable $x_s$ and the objective function is
$\min \sum_{s \in \ov{S}} x_s$.  The set of linear constraints are as
follows: %
\ifshort
$(1)~\forall s\in \ov{S}: x_s \geq 0$, %
$(2)~\forall s \in \wh{S}^*: x_s = \ov{r}(s)$, %
$(3)~\forall s \in \ov{S}_1, (s,t) \in \ov{E}: x_s \geq x_t$, and %
$(4)~\forall s\in \ovSR: x_s = \sum_{t \in \ov{S}} \ov{\trans}(s)(t)
\cdot x_t.$ %
\else
\[
\begin{array}{rcll}
  x_s &\geq &  0 & \quad \forall s\in \ov{S}; \\
  x_s & = & \ov{r}(s) & \quad \forall s \in \wh{S}^*; \\
  x_s & \geq & x_t & \quad \forall s \in \ov{S}_1, (s,t) \in \ov{E}; \\
  x_s & = & \sum_{t \in \ov{S}} \ov{\trans}(s)(t) \cdot x_t & 
  \quad \forall s\in \ovSR.
\end{array} 
\]
\fi The correctness proof of the above linear program to characterize
the value function $\va_{\ov{G}}(\Max{\ov{r}})$ follows by extending
the result for reachability objectives~\cite{FV96}.  The key property
that can be used to prove the correctness of the above claim is as
follows: if a pure memoryless optimal strategy is fixed, then from all
states in $S$, the set $\wh{S}^*$ of absorbing states is reached with
probability~1.  The above property can be proved as follows: since $r$
is a positive reward function, it follows that for all $s \in S$ we
have $\va_G(\MP[p,r])(s)>0$.  Moreover, for all states $s \in S$ we
have $\va_{\ov{G}}(\Max{\ov{r}})(s) =\va_G(\MP[p,r])(s)>0$.  Observe
that for all $s \in S$ we have $\ov{r}(s)=0$.  Hence if we fix a pure
memoryless optimal strategy $\straa$ in $\ov{G}$, then in the Markov
chain $\ov{G}_{\straa}$ there is no closed recurrent set $C$ such that
$C \subseteq S$.  It follows that for all states $s\in S$, in the
Markov $\ov{G}_{\straa}$, the set $\wh{S}^*$ is reached with
probability~1.  Using the above fact and the correctness of
linear-programming for reachability objectives, the correctness proof
of the above linear-program for the objective $\Max{\ov{r}}$ in
$\ov{G}$ can be obtained.  This shows that the value function
$\va_{G}(\MP[p,r])$ for MDPs with reward function $r$ can be computed
in polynomial time. 
\highlight{We can search for a pure memoryless strategy that achieves
  the optimal value by slightly modify the presented procedure. First,
  we check for each end-component if a pure memoryless strategy with
  optimal value exists.  Then, in the transformation to MDP with max
  objective, we create copy states only for states in end-components
  that have optimal pure memoryless strategies. In all states, for
  which the values obtain from the two different transformation to MDP
  with max objective coincide, a pure memoryless strategy that
  achieves the optimal value exists.}  This given us the following lemma.

\begin{lemma}\label{lm:meanpayoff-parity}
  Given a MDP with a mean-payoff parity objective, the value function
  for the mean-payoff parity objective can be computed in polynomial
  time.  \highlight{We can decide in polynomial time if there exists
    a pure memoryless (or finite-state) strategy that achieves the
    optimal value.}
\end{lemma}

Note that, \highlight{in general,} the optimal strategies constructed
for mean-payoff parity requires memory, but the memory requirement is
captured by a counter (which can be represented by a state machine
with state space $\nats$).  The optimal strategy as described in
Lemma~\ref{lemm:key} plays two memoryless strategies, and each
strategy is played a number of steps which can be stored in a counter.
\highlight{Using Lemma~\ref{lemm:key-approx}, we can fix the size the
  counter for any $\epsilon>0$ and obtain a finite-state strategy that
  is $\epsilon$-optimal.
  Lemma~\ref{lemm:memoryless} and the procedure above
  allows us to check in polynomial time if there exists a pure
  memoryless strategy that achieves the optimal value.  This result is
  quite surprising because the related problem of computing the
  optimal pure memoryless strategy, i.e., the strategy that is optimal
  with respect to all pure memoryless strategy is NP-complete; the
  upper bound follows from Theorem~\ref{thm:mc} and the fact that
  emptiness of parity automata can be checked in polynomial
  time~\cite{King01}; the lower bound follows from a reduction of the
  directed subgraph homeomorphism problem~\cite{FHW80}.  }

Lemma~\ref{lm:mdp} and
Lemma~\ref{lm:meanpayoff-parity} yield the following theorem.

\begin{theorem}
  Given a Parity specification $A$, a Mean-payoff specification $B$,
  and a labeled Markov chain $(G,\lambda)$ defining a probability
  measure $\Prb$ on $(\A_I^\omega,{\cal F})$, we can construct a state
  machine $M$ (if one exists) in polynomial time that satisfies $L_A$
  under $\mu$ and optimizes $L_B$ under $\mu$.  \highlight{ We can
    decide in polynomial time if~$M$ can be implemented by a
    finite-state machine.  If~$M$ requires infinite memory, then for
    all $\epsilon>0$, we can construct a finite-state machine~$M'$
    that satisfies $L_A$ under $\mu$ and optimizes $L_B$ under $\mu$
    within~$\epsilon$.  }
\end{theorem}

\section{Experimental Results}
\label{sec:experiments}
In this section we illustrate which types of systems, we can construct
using qualitative and quantitative specifications under probabilistic
environment assumptions.  
\highlight{We have implemented the approach as part of \quasy, our quantitative
synthesis tool~\cite{quasy}.  Our tool takes qualitative and
quantitative specifications and automatically constructs a system that
satisfies the qualitative specification and optimizes the quantitative
specification, if such a system exists. The user can choose between a
system that satisfies and optimizes the specifications (a) under all
possible environment behaviors or (b) under the most-likely
environment behaviors given as a probability distribution on the
possible input sequences.

We are interested in the latter functionality, i.e., in systems that
are optimal for the average-case behaviors of the environment.  In
this case, a specification consists of (i) a safety or a parity
automaton $A$, (ii) a mean-payoff automaton~$B$, and an environment
assumption~$\mu$, given as a set of probability distributions $d_s$
over input letters for each state $s$ of~$B$.
Our implemenation first builds the product of~$A$ and~$B$. Then, it construct the
corresponding MDP $G$. If~$A$ is a safety specification, our implementation
computes an optimal pure memoryless strategy using policy iteration
for multi-chain MDPs~\cite{FV96}.  Finally, if the value of the
strategy is different from~$\bot$, then it converts the strategy to a
finite-state machine $M$ which satisfies $L_A$ (under $\mu$) and is
optimal for $B$ under $\mu$. 
In the case of parity specifications, we implemented the algorithm
described in Section~\ref{sec:meanpayoffparity}. Then, our implementation produces
two mealy machines $M_1$ and $M_2$ as output: (i)~$M_1$ is optimal wrt
the mean-payoff objective and (ii)~$M_2$ almost-surely satisfies the
parity objective. The actual system corresponds to a combination of
the two mealy machines based on inputs from the environment switching
over from one mealy machine to another based on a counter as explained
in Section~\ref{sec:meanpayoffparity}.
More precisely, if we use the strategy used in the proof of
Lemma~\ref{lemm:key}, we obtain an optimal but infinite-state system,
because the size of the counter cannot be bounded.  If we aim for a
finite-state system, we can use the strategy suggested in proof of
Lemma~\ref{lemm:key-approx} leading to a finite-state system with an
$\epsilon$-optimal value.
Furthermore, Lemma~\ref{lemm:memoryless} and the corresponding linear
program in Section~\ref{sec:memoryless} allows us to check if there
exists an optimal pure finite-state strategy. In this case, we can
return a single mealy machine.}

\subsection{Priority-driven Controller.}  In our first experiment, we
took as the quantitative specification~$B$ the product of the
specifications~$A_1$ and~$A_2$ from Example~\ref{ex:controller}
(\figref\ref{fig:grant_fast}), where we sum the weights on the edges.
The qualitative specification is a safety automaton $A$ ensuring
mutually exclusive grants. We assumed the constant probabilities
$P(\set{r_1=1})=0.4$ and $P(\set{r_2=1})=0.3$ for the events $r_1=1$
and $r_2=1$, respectively.  The optimal machine constructed by the
tool is shown in \figref\ref{fig:g2M3}. 
Note that its behavior does not depend on the state, i.e., State~$q_0$
and~$q_1$ are simulation equivalent and can be merged. Since our tool
does not minimize state machines yet, we obtain a system with two
states.
This system behaves like a
priority-driven scheduler. It always grants the resource to the client
that is more likely to send requests, if she is requesting it.
Otherwise, the resource is granted to the other client.  Intuitively,
this is optimal because Client~1 is more likely to send requests and
so missing a request from Client~2 is better than missing a request
from Client~1.  

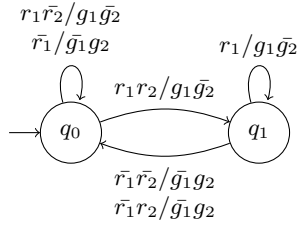
\begin{figure}[t]
\begin{center}
\begin{tikzpicture}[node distance=2.5cm,auto,bend angle=20]
  {\fsize
    \node[state, initial, initial text=] (n2) {$q_0$};
    \node[state, right of=n2] (n3) {$q_1$};
   
    \path[->] (n2)   
    edge [loop above] node {
      $\begin{array}{c}
        r_1 \bar{r_2}/g_1 \bar{g_2}\\
        \bar{r_1}/\bar{g_1} g_2
      \end{array}$} (n2)
    edge [bend left]    node {$r_1 r_2/g_1 \bar{g_2}$}  (n3);
   
    \path[->] (n3)
    edge [loop above] node {$r_1/g_1 \bar{g_2}$} (n3)
    edge [bend left] node {
      $\begin{array}{c}
        \bar{r_1} \bar{r_2}/\bar{g_1} g_2\\
        \bar{r_1} r_2/\bar{g_1} g_2
      \end{array}$
    } (n2);
  }
\end{tikzpicture}

\end{center}
  \caption{Optimal Mealy machine for the $2$-client specification
    without response constraints and the safety automaton $G_2$}
  \label{fig:g2M3}
\end{figure}

\begin{table}[t]
  \caption{Results for 2 to 7 clients without response
    constraints\label{tab:cavimproved}}
\begin{center}
  {\begin{tabular}{cccccc}
      \toprule
     Clients & States in $A\times B$ & States in $G$ & States in $M$
      & Value of $M$ & Time in s\\
      \midrule
      2 & 4 & 13 & 2 & 1.854 & 0.50 \\
      3 & 8 & 35 & 4 & 2.368 & 0.81 \\
      4 & 16 & 97 & 8 & 2.520 & 1.64 \\
      5 & 32 & 275 & 16 & 2.534 & 3.43 \\
      6 & 64 & 793 & 32 & 2.534 & 15.89 \\
      7 & 128 & 2315 & 64 & 2.534 & 34.28\\
     \bottomrule
   \end{tabular}}
\end{center}
\end{table}

\begin{table}[t]
  \caption{Results for 2 to 4 clients with response
    constraints\label{tab:gen2}}
\begin{center}
  {\begin{tabular}{ccccc}
      \toprule
     Clients & States in $A\times B$ & States in $G$ & States in
      $M$ & Value of $M$\\
      \midrule
      2 &	3 &	11 &	3 &	1.850\\
      3 &	34 &	156 &	16 &	2.329\\
      4 &	125 &	557 &	125 &	2.366\\
      \bottomrule
    \end{tabular}
  }
\end{center}
\end{table}

\subsection{Fair Controller.}  In the second experiment, we
added response constraints to the safety specification.
The constraints are given as safety automata that require that every
request is granted within two steps.  We added one automaton $C_i$ for
each client~$i$ and the final qualitative specification was $A\times
C_1\times C_2$.
The optimal machine the tool constructs is System~$M_2$ of
Example~\ref{ex:controller} (\figref\ref{fig:system2}).  System~$M_2$
follows the request sent, if only a single request is send.  If both
clients request simultaneously, it alternates between $g_1$ and $g_2$.
If none of the clients is requesting it grants $g_1$.  Recall that
system $M_1$ and $M_2$ from Example~\ref{ex:controller} exhibit the
same worst-case behavior, so a synthesis approach based on optimizing
the worst-case behavior would not be able to construct~$M_2$.

\subsection{General Controllers.}  We reran both experiments for
several clients.  Again, the quantitative specification was the
product of $A_i$'s.  We used a skewed probability distribution with
$P(\set{r_n =1})= 0.3$ and $P(\set{r_i=1}) = P(\set{r_{i+1}=1}) + 0.1$
for $1\le i\le6$
and the qualitative specification required mutual exclusion.
Table~\ref{tab:cavimproved} shows in the first three columns the
number of clients, the size of the specification ($A\times B$), and
the size of the corresponding MDP $(G)$.  Column~4 and~5 show the size
and the value of the resulting machine $(M)$, respectively. The last
column shows the time needed to construct the system. The runs took
between half a second and half a minute. The systems generated as a
result of this experiment have an intrinsic priority to granting
requests in order of probabilities from largest to smallest.
Table~\ref{tab:gen2} shows the results when adding response
constraints that require that every request has to be granted within
the next $n$ steps, where $n$ is the number of clients. This
experiment leads to quite intelligent systems which prioritize with
the most probable input request but slowly the priority shifts to the
next request variable cyclically resulting into servicing any request
in $n$ steps when there are $n$ clients.  Note that these systems are
(as expected) quite a bit larger than the corresponding
priority-driven controllers.

\section{Conclusions and Future Work}
\label{sec:conclusion}

In this paper we showed how to measure and synthesize systems under
probabilistic environment assumptions wrt qualitative and quantitative
specifications.  We considered the satisfaction of the qualitative
specification with probability~1 ($M \models_{\Prb} \varphi$).
Alternatively, we could have considered the satisfaction of the
qualitative specification with certainty ($M \models \varphi$).  For
safety specification the two notions coincide, however, they are
different for parity specification.  The notion of satisfaction of the
parity specification with certainty and optimizing the mean-payoff
specification can be obtained similar to the solution of mean-payoff
parity games~\cite{CHJ05} by replacing the solution of mean-payoff
games by solution of MDPs with mean-payoff objectives.  However, since
solving MDPs with parity specification for certainty is equivalent to
solving two-player parity games, and no polynomial time algorithm is
known for parity games, the algorithmic solution for the satisfaction
of the qualitative specification with certainty is computationally
expensive as compared to the polynomial time algorithm for MDPs with
mean-payoff parity objectives.  Moreover, under probabilistic
assumption satisfaction with probability~1 is the natural notion.

\highlight{
We have implemented our algorithm in the tool \quasy, a quantitative
synthesis tool for constructing worst-case and average-case optimal
systems with respect to a qualitative and a quantitative
specification. We can check if an optimal finite-state system
exists and then either constructs an optimal or an $\epsilon$-optimal
system depending on the outcome of the check.

In our future work, we will explore different directions to improve
the performance of \quasy.  In particular, a recent paper by
Wimmer et al.~\cite{wimmer10} presents an efficient technique for solving MDP with
mean-payoff objectives based on combining symbolic and explicit
computation.  We will investigate if symbolic and explicit
computations can be combined for MDPs with mean-payoff parity
objectives as well.}

\bibliographystyle{plain}
\bibliography{main,tool}
\end{document}